\documentclass[11pt]{article}
\usepackage[ruled,linesnumbered]{algorithm2e} 
\usepackage{amsmath,amsthm,amssymb,color}
\usepackage{xcolor}
\usepackage{latexsym,bbm,xspace,graphicx,float}
\definecolor{newblue}{rgb}{0.19, 0.55, 0.91}
\usepackage[colorlinks,citecolor=newblue,bookmarks=true,pagebackref=true, urlcolor=blue, linkcolor=blue, linktoc=page]{hyperref}
\usepackage[capitalise, nameinlink]{cleveref}
\SetKwRepeat{Do}{do}{while}
\usepackage[utf8]{inputenc} 
\usepackage[T1]{fontenc}    
\usepackage{hyperref}       
\usepackage{url}            
\usepackage{booktabs}       
\usepackage{amsfonts}       
\usepackage{nicefrac}       
\usepackage{microtype}      
\usepackage{xcolor}         

\usepackage{cs270}

\SetCommentSty{mycommfont}
\usepackage{comment}

\newcommand{\polylog}[1]{\textnormal{polylog}\,{#1}\xspace}
\newcommand{\Ot}{\ensuremath{\widetilde{O}}} 

\begin{document}
\title{Space Complexity of Minimum Cut Problems in Single-Pass Streams}
\date{December 6, 2024}
\author{Matthew Ding\footnote{(\href{mailto:matthewding@berkeley.edu}{matthewding@berkeley.edu}) Department of Electrical Engineering and Computer Sciences, University of California, Berkeley. Supported in part by NSF CCF-1951384.}
\and
Alexandro Garces\footnote{(\href{mailto:agarces2@mit.edu}{agarces2@mit.edu}) Department of Mathematics, Massachusetts Institute of Technology. Supported in part by NSF CNS-2150186.}
\and
Jason Li\footnote{(\href{mailto:jmli@alumni.cmu.edu}{jmli@cs.cmu.edu}) Computer Science Department, Carnegie Mellon University. This work was done in part as a Research Fellow at the Simons Institute for the Theory of Computing.}
\and
Honghao Lin\footnote{(\href{mailto:honghaol@andrew.cmu.edu}{honghaol@andrew.cmu.edu}) Computer Science Department, Carnegie Mellon University. Supported in part by a Simons Investigator Award, NSF CCF-2335412, and a CMU Paul and James Wang Sercomm Presidential Graduate Fellowship.} 
\and
Jelani Nelson\footnote{(\href{mailto:minilek@berkeley.edu}{minilek@berkeley.edu}) Department of Electrical Engineering and Computer Sciences, University of California, Berkeley. Supported in part by NSF CCF-1951384 and NSF CCF-2427808.}
\and
Vihan Shah
\footnote{(\href{mailto:vihan.shah@uwaterloo.ca}{\text{vihan.shah@uwaterloo.ca}}) Department of Computer Science, University of Waterloo. Supported in part by Sepehr Assadi's Sloan Research Fellowship and NSERC Discovery Grant.} 
\and
David P. Woodruff\footnote{(\href{mailto:dwoodruf@andrew.cmu.edu}{dwoodruf@andrew.cmu.edu}) Computer Science Department, Carnegie Mellon University. Supported in part by a Simons Investigator Award and NSF CCF-2335412.}
}

\clearpage\maketitle
\thispagestyle{empty}

\begin{abstract}
     We consider the problem of finding a minimum cut of a weighted graph presented as a single-pass stream. While graph sparsification in streams has been intensively studied, the specific application of finding minimum cuts in streams is less well-studied. To this end, we show upper and lower bounds on minimum cut problems in insertion-only streams for a variety of settings, including for both randomized and deterministic algorithms, for both arbitrary and random order streams, and for both approximate and exact algorithms. One of our main results is an $\widetilde{O}(n/\varepsilon)$ space algorithm with fast update time for approximating a spectral cut query with high probability on a stream given in an arbitrary order. Our result breaks the $\Omega(n/\varepsilon^2)$ space lower bound required of a sparsifier that approximates all cuts simultaneously. Using this result, we provide streaming algorithms with near optimal space of $\widetilde{O}(n/\varepsilon)$ for minimum cut and approximate all-pairs effective resistances, with matching space lower-bounds. The amortized update time of our algorithms is $\widetilde{O}(1)$, provided that the number of edges in the input graph is at least $(n/\varepsilon^2)^{1+o(1)}$. We also give a generic way of incorporating sketching into a recursive contraction algorithm to improve the post-processing time of our algorithms. In addition to these results, we give a random-order streaming algorithm that computes the {\it exact} minimum cut on a simple, unweighted graph using $\widetilde{O}(n)$ space. Finally, we give an $\Omega(n/\varepsilon^2)$ space lower bound for deterministic minimum cut algorithms which matches the best-known upper bound up to polylogarithmic factors.  
\end{abstract}

\thispagestyle{empty}
\newpage
\thispagestyle{empty}
\tableofcontents
\newpage
\clearpage
\setcounter{page}{1}

\section{Introduction}
We consider the graph streaming model, which is a key model for computations on massive graph datasets that has been extensively studied over the past couple of decades (see, e.g., \cite{McGregor2014} for a survey).
We specifically study the problem of finding minimum cuts in adversarial and random-order streams, which has been less studied than the problem of general cut sparsification of graphs. Besides being theoretically interesting, finding minimum cuts is also a problem of practical interest. For example, they allow for the calculation of social network metrics such as influence \cite{wang2020efficient} and are used to quantify the robustness of power networks \cite{HendrickxPowerNetwork} and road networks \cite{SugiuraRoadNetwork} to failures. 

\subsection{Adversarial Streams}

The key method so far for solving minimum cut with low-memory is the usage of cut sparsifiers. The notion of cut sparsifiers was introduced by Benczur and Karger~\cite{BK96} and has been extremely influential. Given a weighted graph $G = (V, E, w)$ with $n = |V|$ vertices and $m = |E|$ edges, and polynomially bounded edge weights $w:E \rightarrow \mathbb{R}_+$, together with an accuracy parameter $\eps > 0$, a cut sparsifier of $G$ is a sparse subgraph $H$ on the same vertex set $V$ but with (possibly) different edge weights such that the weight of every cut in $G$ is $(1+\eps)$-approximated by the weight of the corresponding cut in $H$. For two sets $S, T \subseteq V$, we use $E(S, T) = \{(u, v) \in E: u \in S, v \in T\}$ to denote the set of edges between $S$ and $T$ in graph $G$ and $w_G(S, T) = \sum_{e \in E(S, T)} w_e$ to denote the total weight of edges between $S$ and $T$ in graph $G$.  Formally, we have the following definition,

\begin{definition}[For-All Cut Sparsifier]
$H$ is a $(1 + \eps)$ for-all cut sparsifier of $G$ if and only if the following holds for all $\varnothing \subset S \subset V$:
\[w_H(S, V \setminus S) \approx_\eps w_G(S, V \setminus S)\]
where $a \approx_\eps b$ is defined as $(1-\eps)\cdot b \leq a \leq (1+\eps)\cdot b$.
\end{definition}

Specifically, \cite{BK96} shows that a cut sparsifier always exists with $O(n \log n/\eps^2)$ edges. This bound was improved to $O(n/\eps^2)$ edges by \cite{BSS12} and \cite{alon1997edge,AndoniCKQWZ16,carlson2019optimal} proved a lower bound of $\Omega(n(\log n)/\eps^2)$ bits (even in the case when the cut-sparsifier $H$ is a not
necessarily a subgraph of $G$).
The existence bound was also extended to the stronger notion of spectral sparsifiers (\cite{st04, ST11, SS11, BSS12}), where the quadratic form associated with the Laplacian of graph $H$ provides a $(1 + \eps)$-approximation to that of $G$. Let $L_G$ denote the Laplacian matrix of graph $G$. For every subset $S \subseteq V$, let $x_S$ be the binary indicator vector of $S$ such that $x_i = 1$ if $i \in S$ and $x_i = 0$ if $i \notin S$. Then we have $x_S^ \top L_{G} x_S = w_G(S, V \setminus S)$. The definition of the spectral sparsifier extends the assumption that $x$ is a binary vector to an arbitrary vector.

\begin{definition}[For-All Spectral Sparsifier]
Let $G$ and $H$ be two weighted undirected graphs. Fix $0 < \eps < 1$. We say $H$ is a $(1 + \eps)$ for-all spectral sparsifier of $G$ iff the following holds:
\[(1-\eps)\cdot L_G \preceq L_H \preceq (1+\eps) \cdot L_G\]
where we use $A \preceq B$ to denote $\forall x \in \R^{n}, x^\top A x \leq x^\top B x$.
\end{definition}
This construction has had a tremendous impact on cut problems in graphs, see, e.g., \cite{BK96, BK02, KL02, She09, Mad10}.
However, for very small values of $\eps$, the $1/\eps^2$ dependence in cut sparsifiers may be prohibitive on large-scale graphs. Motivated by this, the work of \cite{AndoniCKQWZ16} relaxed the cut sparsification problem to that of outputting a data structure $\mathrm{sk}_G$ such that given any {\it fixed} cut $S \subseteq V $, the value of $\mathrm{sk}_G(S)$ is within a $(1 + \eps)$ factor of the cut value of $S$ in $G$, with probability at least $2/3$\footnote{This can be amplified to high probability, which we define to mean with probability at least $1 - 1/\mathrm{poly}(n)$. This is done by independently repeating the data structure $O(\log n)$ times and outputting the median estimate. This also lets us estimate the cut value for $\textrm{poly}(n)$ cuts simultaneously.}. This relaxed notion of the problem is called the ``for-each'' model, which should be contrasted with the previous ``for-all'' model of cut and spectral sparsification. 

Surprisingly, \cite{AndoniCKQWZ16} showed that such a data structure exists for cut sparsification with $\mathrm{poly}(n)$-bounded integer edge weights of size $\widetilde{O}(n/\eps)$\footnote{Throughout we use $\widetilde{O}(\cdot)$ to hide $\polylog(n)$ factors.}, which is optimal up to polylogarithmic factors. In a follow-up work of~\cite{jambulapati2018efficient}, the authors extended the $\widetilde{O}(n/\eps)$ upper bound to the for-each model for spectral sparsifiers. Another very interesting work is ~\cite{CGP+18}, which shows that such a sketch can be chosen as a reweighted subgraph of $G$. We specifically deal with these graphical for-each sparsifiers:
\begin{definition}[Graphical For-Each Spectral Sparsifier]
Let $G$ and $H$ be two weighted undirected graphs. Fix $0 < \eps < 1$. We say $H$ is a $(1 + \eps)$ graphical for-each spectral sparsifier of $G$ if and only if for each $x \in \R^n$, with probability at least $2/3$:
\[x^\top L_H x \approx_\eps x^\top L_G x\]
\end{definition}

In a line of work, see, e.g., \cite{McGregor2014, KLM+17, spectral_stream_2020}, efficient algorithms for sparsifiers in the for-all model in a graph stream were constructed. The state-of-the-art work~\cite{spectral_stream_2020} gives a single-pass algorithm for spectral sparsification in dynamic streams which uses space $\widetilde{O}(n/\eps^2)$. 
\begin{lemma} [\cite{spectral_stream_2020}]
\label{lem:for_all}
    Given an input graph $G$ in a stream, there exists a one-pass streaming algorithm that outputs a subgraph $H$ of $G$ such that with high probability $H$ is a $(1 + \eps)$-spectral sparsifier of $G$. Moreover, the algorithm uses  $\Ot(n/\eps^2)$ space and $\Ot(m + n/\eps^2)$ time. 
\end{lemma}

However, it was not known how to construct a for-each sketch for cut sparsification in a stream in better than $\widetilde{O}(n/\eps^2)$ space, which follows from computing a for-all sparsifier. A natural question is whether we can implement a streaming algorithm that computes a for-each sparsifier with space matching the $\Omega(n/\eps)$ offline lower bound that holds for for-each sparsification \cite{AndoniCKQWZ16}. 

If we can answer this question in the affirmative, we can also get a better space bound for the approximate minimum cut problem.
It is well-known that there are only poly$(n)$ $O(1)$-approximate minimum cuts\footnote{When we say a cut is an $\alpha$-approximate minimum cut, we mean that its cut value is at most $\alpha$ times the minimum cut value.} in a graph. So we can run a for-all cut sparsification algorithm in a stream with $\eps = \Theta(1)$ in $\widetilde{O}(n)$ space and obtain a candidate set of poly$(n)$ cuts containing the minimum cut. In parallel, we estimate the value of each cut using a for-each sparsifier with failure probability $1/\textrm{poly}(n)$. 
The smallest cut out of the candidate cuts will be a $(1+\eps)$-approximate minimum cut. Thus, the space complexity is just $\widetilde{O}(n)$ plus that of maintaining a for-each sparsifier in a stream. However, the existing algorithms for computing for-each sparsifiers are not streaming algorithms. Motivated by the existence of a for-each sketch in $\widetilde{O}(n/\eps)$ bits of space, one may wonder if it is possible to obtain a $(1 + \eps)$-approximation to the minimum cut in $\widetilde{O}(n/\eps)$ space in a single-pass stream. 

\begin{question}\label{question:adversarial}
    \em What is the space and time complexity of obtaining a $(1 + \eps)$-approximate minimum cut in a single-pass insertion-only stream?
\end{question}

We also note that an exact global minimum cut streaming algorithm exists in $\widetilde{O}(n)$ space with two passes for simple graphs \cite{two_pass} (see also \cite{rsw18}) and $\log n$ passes for weighted graphs \cite{MN20}.

\subsection{Random-order streams}

We know that finding the minimum cut exactly in single-pass adversarial streams needs $\Omega(n^2)$ space \cite{Zelke2011Intractability}. We need an additional pass to get the exact minimum cut in $\Ot(n)$ space. 
However, many streaming applications require solving the problem in exactly one pass because the stream cannot be stored. 

Thus, to surpass the $\Omega(n^2)$ barrier of adversarial streams in one pass, we consider a relaxation known as random-order streams and ask whether we can find the minimum cut exactly in $o(n^2)$ space in this setting.
In this model, the graph can still be chosen adversarially, but its edges arrive in the stream in a random order.

The motivation for studying this model stems from the fact that real-world data is typically not adversarial. Adversarial streams are adversarial in both their input and arrival order, while random-order streams relax the arrival order, making them more representative of real-world data.
Additionally, the random-order streaming model may help explain the empirical performance of certain heuristics.

This model was first studied by \cite{MUNRO1980315} for sorting in limited space. Many subsequent works have studied problems in random-order streams \cite{doi:10.1137/07069328X,kapralov2014approximating,robust_2016,monemizadeh2017testable,peng2018estimating,czumaj2019testable,chakrabarti2020vertex}, some of which show a clear separation between the adversarial and random-order models.

One notable line of work focuses on maximum matching within this model \cite{konrad2012maximum,gamlath2019weighted,konrad2018simple,assadi2019coresets,farhadi2020approximate,bernstein2023improved,assadi_et_al:LIPIcs.ICALP.2021.19,hashemi2024weighted}.
The best-known upper bound in the random order model \cite{assadi_et_al:LIPIcs.ICALP.2021.19} gets a slightly better than $2/3$-approximation in semi-streaming space, which is better than the best possible $\frac{1}{1+\ln 2} \approx 0.59$-approximation in the adversarial model due to the lower bound of \cite{kapralov2021matching}, showing another separation between adversarial and random-order streams.

Motivated by this, we study the minimum cut problem in the random-order model and ask the following question:
\begin{question}\label{question:random}
    \em What is the space complexity of obtaining the exact minimum cut in a single-pass random-order stream?
\end{question}

This work resolves the above questions with matching algorithms and lower bounds. We also carefully optimize the running time and update time of our algorithms. We also obtain several additional results, which we outline in Section \ref{sec:results}.

\section{Preliminaries}

\subsection{All-Pairs Effective Resistances}
Given $G(V, E, w)$ as an electrical network on $n$ nodes in which each edge $e$ corresponds to a link of conductance $w_e$ (i.e., a resistor of resistance $1/w_e$), the effective resistance of an edge $e$ is the potential difference induced across it when a unit current is injected at one end of $e$ and extracted at the other end of $e$.
Equivalently it is equal to $R_{\text{eff}}^G = \chi^\top_{u,v} L^\dag_G \chi_{u,v}$, where $L^\dag_G$ denotes the pseudoinverse of the Laplacian of graph $G$ and $\chi_{u,v} := 1_u - 1_v$. 

In the all-pairs effective resistance problem, we are required to output a data structure, which with high probability can generate effective resistance values for each of the $n^2$ pairs of vertices, or a $(1 + \eps)$-approximation of each of the $n^2$ pairs of values in the $(1+\eps)$-approximate version of the problem. Note that the data structure implicitly represents these $n^2$ values and responds with one such value on a given query pair.

\subsection{Short-Cycle Decomposition}
Short cycle decompositions are a recent algorithmic tool for graph sketching problems introduced by Chu et al.\ \cite{CGP+18}. 

\begin{definition}[Short-Cycle Decomposition]
    An $(\hat{m}, L)$ short-cycle decomposition of a graph $G$ is a decomposition of the graph into edge-disjoint cycles of at most length $L$ and an additional $\hat{m}$ edges outside of these cycles.
\end{definition}

The main contribution of \cite{CGP+18} to spectral sparsification is in proving the following claim:
\begin{claim}\label{claim:graphical_sparsifier}
    Given an undirected weighted graph $G$ and an $(\hat{m}, L)$ short-cycle decomposition routine \textsc{CycleDecomp}, there exists an algorithm which returns with high probability a $(1+\eps)$-spectral sparsifier \emph{graph} with $\Ot(\hat{m} + nL\eps^{-1})$ edges and run-time
    \begin{equation*}
        \Ot(m) + T_{\textsc{CycleDecomp}}(O(m\log n),n)
    \end{equation*}
    where $T_{\textsc{CycleDecomp}}(m,n)$ is the running time of \textsc{CycleDecomp} on a graph with $m$ edges and $n$ vertices.
\end{claim}

As a high-level overview, the short-cycle decomposition is key as a degree-preserving sparsification method. Given each cycle, we can label them numerically in order and sample either all the odd or all the even labeled edges, each with $1/2$ probability. Previous works on ``for-each'' sparsifiers \cite{AndoniCKQWZ16, jambulapati2018efficient} used a recursive expander decomposition and subsampling at each level, with the main issue being that the degrees were not well-preserved. This issue was solved by explicitly storing the degrees of the graph, causing their spectral sparsifier data structure to no longer be a graph. The work of \cite{CGP+18} circumvents this issue with the short-cycle decomposition by constructing subsampled graphs that exactly preserve their vertex degree.

Note that \cite{CGP+18} already gives a graphical spectral sparsifier with $\Ot(n/\eps)$ edges and bits of working memory using a basic brute-force search for cycle decomposition. However, its runtime for construction is $O(mn)$, which is prohibitively slow as a subroutine. On the other hand, later work by Parter and Yogev \cite{optimal_cycle_decomposition} shows a deterministic algorithm to compute $(O(n\log n), O(\log^2n))$ short-cycle decompositions. This combined with \cref{claim:graphical_sparsifier} will allow us to design an algorithm in $m^{1+o(1)}$ time which returns a $(1+\eps)$-spectral sparsifier which is a reweighted graph with $\Ot(n/\eps)$ edges. This gives near-optimal parameters except for a required $m^{1+o(1)}$ bits of working memory. However, in \cref{sec:graphical_spectral_sketch}, we show that we can implement this subroutine so that it is guaranteed to use only $\Ot(m)$ working memory, achieving near-optimal results in both working memory usage and running time.

\section{Overview of Results}\label{sec:results}
\begin{table}
\caption{Minimum Cut Space Complexity in the Single-Pass Insertion-Only Streaming Setting}
\renewcommand{\arraystretch}{1.5}  
\begin{center}
\begin{tabular}{||l | l | l | l||} 
 \hline
 Stream Type & Exact/Approx $(1+\eps)$ & Upper Bound & Lower Bound\\ [0.5ex] 
 \hline\hline
 Adversarial & Exact & $O(n^2)$ (full graph) & $\Omega(n^2)$ \cite{Zelke2011Intractability} \\ 
 \hline
 Adversarial & Approx, Deterministic & $\Ot(n/\eps^2)$ \cite{BSS12} & $\Omega(n/\eps^2)^*$ (\cref{thm:lower_bound})\\
    \hline
 Adversarial & Approx, Randomized & $\widetilde{O}(n/\eps)$ (\cref{corollary:min_cut}) & $\Omega(n/\eps)$ (\cref{thm:lower_bound}) \\
 \hline
 Random-Order & Exact & $\widetilde{O}(n)$ (\cref{thm:random_order}) & $\Omega(n)$ \cite{robust_2016} \\ [1ex] 
 \hline
\end{tabular}
\end{center}
\label{table:minimum_cut}
\end{table}

We list the current state-of-the-art results for single-pass minimum cut streaming algorithms in \cref{table:minimum_cut}. Our main result is an $\widetilde{O}(n/\eps)$ space streaming algorithm that finds a $(1+\eps)$-approximate minimum cut in a single-pass stream. 

\begin{theorem}\label{corollary:min_cut}
    There is a one-pass insertion-only streaming algorithm that, with high probability, computes an $(1+\eps)$-approximation of the minimum cut on weighted graphs using $\widetilde{O}(n/\eps)$ bits of space. Moreover, our algorithm takes $\widetilde{O}(m) + (n/\eps^2)^{1 + o(1)}$ total update time and $\widetilde{O}(n^2/\eps^2)$ post-processing time.
\end{theorem}

This shows, somewhat surprisingly, that estimating the minimum cut is easier than computing a for-all cut sparsifier in a data stream. Our algorithm is mainly based on the following new for-each spectral sparsifier in a graph stream.

\begin{lemma}\label{thm:sparsifier}
    A one-pass insertion-only streaming algorithm exists that, with high probability, constructs a $(1+\eps)$ for-each spectral sparsifier of weighted graphs with $\widetilde{O}(n/\eps)$ edges using $\widetilde{O}(n/\eps)$ bits of space. This algorithm also has total runtime $\widetilde{O}(m) + (n/\eps^2)^{1 + o(1)}$.
\end{lemma}

Our algorithm's total runtime is $\widetilde{O}(m) + (n/\eps^2)^{1 + o(1)}$, which notably implies its amortized update time is $\widetilde{O}(1)$ when $m \ge (n/\eps^2)^{1 + o(1)}$. With the $\Omega(n/\eps)$ data structure lower bound in~\cite{AndoniCKQWZ16}, our algorithm is tight in space complexity up to polylogarithmic factors. Based on our \cref{thm:sparsifier}, we can obtain an $\widetilde{O}(n/\eps)$ space streaming algorithm that can $(1 \pm \eps)$-estimate all-pairs effective resistances.

\begin{corollary}\label{corollary:efffective_resistance}
There exists a one-pass insertion-only streaming algorithm that constructs a data structure that calculates $(1+\eps)$-approximations to the effective resistances between every pair of vertices on weighted graphs using $\widetilde{O}(n/\eps)$ bits of space. Our algorithm has $\widetilde{O}(m) + (n/\eps^2)^{1 + o(1)}$ total time during the stream and $\widetilde{O}(n^2/\eps)$ post-processing time.
\end{corollary}

Notice that we are able to match the sketching bounds of \cite{jambulapati2018efficient} up to polylogarithmic factors in the space and time required.

We also show an $\Omega(n/\eps)$ lower bound for minimum cut and all-pairs effective resistances in \cref{sec:min_cut} and \cref{sec:resistance} respectively, showing that our results of \cref{corollary:min_cut} and \cref{corollary:efffective_resistance} are space-optimal up to polylogarithmic factors.
\begin{theorem}\label{thm:lower_bound}
    Fix $\eps > 1/n$. Any randomized algorithm that outputs a $(1+\eps)$-approximation to the minimum cut of a simple, undirected graph in a single pass over a stream with probability at least $2/3$ requires $\Omega(n/\eps)$ bits of space. If the algorithm is deterministic and $\eps \geq 1/n^{1/4}$, then the algorithm requires $\Omega(n/\eps^2)$ bits of space.\footnote{This latter assumption on $\eps$ is the same as used in \cite{carlson2019optimal}.} 
\end{theorem}

\begin{theorem}
\label{thm:lower_bound_effective_resistance}
Fix $\eps > 1/n$. Suppose $sk(\cdot)$ is a sketching algorithm that
outputs at most $s = s(n, \eps)$ bits, and $f$ is an estimation algorithm such that, 
\[
\forall a, b \in V, \ \ \ \ \ \Pr[f(a, b, sk(G)) \in (1 \pm \eps)r_{a, b}] \ge \frac{2}{3}, 
\]
where $r_{a, b}$ is the effective resistance of nodes $a, b$. Then we have $s \ge \Omega(n/\eps)$.
\end{theorem}

We also study the minimum cut problem in the random-order streaming model, where edges arrive in a random order instead of an arbitrary worst-case order. We prove the following result for \emph{simple} unweighted graphs, which are graphs with at most one edge between any two vertices.
\begin{theorem}\label{thm:random_order}
    There exists a one-pass insertion-only streaming algorithm in the random-order model that outputs all minimum cuts $(S, V\setminus S)$, along with their corresponding edges $E(S, V\setminus S)$, for a simple, unweighted graph and with high probability using $\widetilde{O}(n)$ space. The algorithm has $\widetilde{O}(n)$ update time and $\widetilde{O}(n^2 k)$ post-processing time, where $k$ is the value of the minimum cut.
\end{theorem}
We note that \cref{thm:random_order} provides the exact minimum cut, which we find surprisingly possible in a single pass. This further adds to the surprises regarding minimum cut in the streaming model, as a $2$-pass exact algorithm in arbitrary order streams was known \cite{two_pass} while computing the exact minimum cut in a single pass in an arbitrary order stream is known to require $\Omega(n^2)$ memory \cite{Zelke2011Intractability}. 

Given that determining whether a graph is connected from a random order stream requires $\Omega(n)$ space \cite{robust_2016}, our $\widetilde{O}(n)$ space algorithm for finding the exact minimum cut in random-order streams is optimal up to polylogarithmic space factors, even for just outputting the value. Note that we compute the minimum cut and all edges crossing the cut in addition to the value. 

\subsection{Our Techniques}
\paragraph{Approximate Minimum Cut.}
Suppose that $H$ is a $(1 + \eps)$-for-all sparsifier of $G$. Then the minimum cut of $H$ is a $(1 + \eps)$-approximation to that of $G$. However, such a conclusion will not hold if $H$ is instead a $(1+\eps)$ for-each sparsifier, as the total number of cuts is exponential. As discussed earlier, fortunately the number of $1.1$-approximate minimum cuts is $O(n^2)$ (\cite{karger00}). Hence, if we run a for-all sparsifier algorithm in parallel with accuracy $\eps' = O(1)$, we can obtain a list of $O(n^2)$ candidate minimum cuts in $G$. Using the the $(1 + \eps)$ for-each sparsifier, we can then find a $(1 + \eps)$-approximation to the minimum cut by union bounding over the $O(n^2)$ candidates.

For the lower bound, we consider the $k$-edge-connectivity problem. The \emph{edge-connectivity} of a graph is the minimum number of edges that need to be deleted to disconnect the graph. The $k$-edge-connectivity problem asks whether the edge connectivity of a graph is $<k$ or $\geq k$. Suppose that there is an algorithm that solves the $(1 + \eps)$-approximate minimum cut value problem. Then, we can use it to solve the $k$-edge-connectivity problem with $k = O(1/\eps)$. Combining with the $\Omega(kn)$ bit space lower bound in the work of~\cite{SW15}, it follows that there is an $\Omega(n/\eps)$ bit space lower bound.

\paragraph{For-Each Spectral Sparsifier.} 
The main difficulty in implementing the algorithms of~\cite{AndoniCKQWZ16} and~\cite{jambulapati2018efficient} in the streaming setting is that both algorithms require careful graph decomposition. Take the algorithm in~\cite{AndoniCKQWZ16} as an example: it partitions the graph into several components, where each component is well-connected and does not have a sparse cut smaller than $1/\eps$. In the worst case, such a procedure may take $\log n$ levels recursively, which seems unachievable in a single-pass stream. Besides, since the for-each sketch of both algorithms is not a graph, it is unclear how to merge two sketches directly. Indeed, natural ways of merging such sketches may destroy the decomposition into sparse cuts. To address these issues, we instead consider the work of~\cite{CGP+18} in which the authors show how to construct a for-each sparsifier which is, in fact, a re-weighted subgraph of the input graph, with $n^{1 + o(1)}/\eps$ edges. Combining this and the merge-and-reduce framework gives a for-each spectral sparsifier with $n^{1 + o(1)}/\eps$ edges in a single-pass stream. The extra $n^{o(1)}$ factor in the space is an issue here. We notice that the extra $n^{o(1)}$ factor in the work of~\cite{CGP+18} is due to their near-linear time short cycle decomposition algorithm. Leveraging ideas from~\cite{optimal_cycle_decomposition}, we give a new short-cycle decomposition algorithm that trades off running time for space, namely, it achieves $m^{1 + o(1)}$ time and $\widetilde{O}(m)$ space. While this translates to $\widetilde{O}(n/\eps)$ space, it unfortunately gives $n^{o(1)}$ update time per edge. To fix this, we use online leverage score sampling to produce a virtual stream of only $\widetilde{O}(n/\eps^2)$ edges that we instead run our algorithm on. By doing this, we reduce the time to $\widetilde{O}(1)$ per edge provided the number $m$ of edges in the input satisfies $m \geq (n/\eps^2)^{1+o(1)}$.

\paragraph{Approximate All-Pairs Effective Resistance.}
In \cite{jambulapati2018efficient}, the authors show if we can get a for-each sparsifier, then we can use it to generate a data structure in near-linear time, and such a data structure can approximate all pairs of the effective resistance in $\widetilde{O}(n^2/\eps)$ time. Combining this and our \cref{thm:sparsifier} yields our upper bound.

For the lower bound, we use communication complexity and consider a similar graph construction in~\cite{AndoniCKQWZ16}, though we use new arguments based on random walks. Specifically, given a random binary string $s \in \{0, 1\}^{n/\eps}$, we encode it into a graph $G$ where $G$ is divided into $O(\eps n)$ disjoint bipartite graphs $G_i$, and in each $G_i$, the existence of an edge between each pair corresponds to one random bit in $s$. We show that for every such pair $(a, b)$, with high probability there is a $(1 + \eps)$-separation between the two cases (there is an edge between $a, b$ or not), which yields an $\Omega(n/\eps)$ lower bound. To prove this, in particular, we use the connection between the effective resistance and the hitting time on a graph and a recent concentration result about the hitting time on a random graph $G(n, p)$ (each possible edge on $n$ vertices is included independently with probability $p$).

\paragraph{Exact Minimum Cut in Random-order Streams.} Recall that we assume that the graph is simple and unweighted, and each edge arrives in a random order stream. 
We also assume the minimum cut size is $\Omega(\log n)$. If the size is smaller, a for-all cut sparsifier with accuracy parameter $\eps' = 1 / \log^2 n$ gives the exact value of the minimum cut.
The main idea here is that by looking at the prefix of edges in a stream, which we will call the graph $H$, we roughly learn the sizes of all the cuts up to a small constant factor. Using this information, we learn all $1.1$-approximate minimum cuts in $G$. The next question is how to get the exact value of all these cuts (so we can find the minimum one). For the cut edges in $H$, note that at this time, the cut values of all these cuts in $H$ is $\Theta(\log n)$, and hence a for-all sparsifier of $H$ with accuracy parameter $\eps' = 1 / \log^2 n$ gives the exact value of these cuts in $H$. 
We note that the total number of edges that participate in at least one of the $1.1$-approximate non-singleton cuts is $O(n)$ \cite{rsw18}, which is interesting because there could be as many as $O(n^2)$ such cuts. Hence, we can store all these $O(n)$ edges to get the exact cut values of the approximate minimum cuts in $G \setminus H$. Putting the two things together, we obtain the exact value of all $1.1$-approximate minimum cuts and thus obtain the exact minimum cut value in $G$ (along with all the minimum cuts).

\paragraph{Faster Runtime.} 
In the above discussion, our algorithm for $(1 + \eps)$ min-cut in worst-case streams has an $\widetilde{O}(n^3)$ post-processing time, and the algorithm for exact min-cut in a random-order stream has an $\widetilde{O}(n^3)$ update time. The difficulty here is related: when enumerating all the $O(1)$-approximate minimum cuts, we need $O(n)$ time to evaluate the cost for each specific cut. We thus propose a general algorithmic framework to overcome this, which is based on the recursive contraction algorithm in~\cite{KS96} (\cref{sec:min_cut}). Namely, when enumerating all the approximate minimum cuts in the recursion tree, we simultaneously maintain a low-space sketch of the columns of the edge-vertex incidence matrix of the corresponding graph. Specifically, we apply a Johnson-Lindenstrauss sketch in our first case and a sparse recovery sketch in our second case, which reduces the evaluation time $O(n)$ to $O(\log n / \eps^2)$ and $\mathrm{polylog}(n)$, respectively.

\subsection{Open Problems}
While we resolve several gaps between the upper and lower bounds, we discuss a remaining open problem. There exist fully dynamic streaming algorithms for approximate minimum cut using $\Ot(n/\eps^2)$ space where the stream is allowed to both add or delete edges. From \cref{thm:lower_bound}, we know that algorithms solving approximate minimum cut in insertion-only streams must use $\Omega(n/\eps)$ space. Therefore, the space complexity of approximate minimum cut in a dynamic stream lies somewhere within these two values, and we believe resolving this gap is a problem of theoretical interest.

\begin{question}
    \em What is the exact space complexity of calculating a $(1+\eps)$ approximation to the minimum cut of simple weighted graphs in a one-pass dynamic stream?
\end{question}
We do not have an exact conjecture about this, but we provide an alternative and fully self-contained proof to \cref{thm:lower_bound} in \cref{sec:lb}, for which the techniques could help prove a lower bound for the fully dynamic case. 

\section{Spectral Sparsification and Minimum Cut in Worst-Case Streams} \label{sec:worst_case}
A key result of the section is our construction of \cref{thm:sparsifier}, a single-pass insertion-only streaming algorithm for a $(1+\eps)$-for-each cut sparsifier in undirected weighted graphs in $\widetilde{O}(n/\eps)$ space and $n^{o(1)}$ update time.
In \cref{sec:graphical_spectral_sketch} and \ref{sec:online_sampling}, we first give our description of graphical spectral sketches (i.e., spectral sparsifiers that are reweighted subgraphs) and online leverage score sampling, which we use as subroutines in our final algorithm. In \cref{sec:main_algorithm}, we give our full algorithm and prove the correctness of our algorithm. In \cref{sec:min_cut} and \ref{sec:resistance}, we prove the two applications of our algorithm: streaming approximate minimum cut and all-pairs effective resistances.

\subsection{Graphical Spectral Sketches}
\label{sec:graphical_spectral_sketch}
Our algorithm uses the following result in~\cite{CGP+18} as a subroutine.
\begin{lemma}[\cite{CGP+18}]
    \label{lem:spectral_sketch}
    For $\eps \in (0, 1]$, there is an algorithm, which given $G$, runs in time $m^{1 +o(1)}$, and with high probability returns a $(1 + \eps)$-for-each spectral sparsifier with $n^{1 + o(1)}/\eps$ edges.
\end{lemma}

As pointed out by the work of~\cite{optimal_cycle_decomposition}, in the above algorithm, the extra $n^{o(1)}$ factor in the number of edges is due to the nearly-linear time short cycle decomposition used in the whole algorithm. We use the simpler and improved short-cycle decomposition algorithm of \cite{optimal_cycle_decomposition}, which outputs in $m^{1+o(1)}$ time a collection of edge-disjoint cycles of length $O(\log^2n)$ that cover all but $O(n\log n)$ edges of the graph. Na\"ively, the algorithm is also implemented in $m^{1+o(1)}$ space, but we show that a simple modification achieves near-linear space, as follows:
\begin{lemma}
\label{lem:cycle_decomposition}
There is an $m^{1+o(1)}$ time, $\Ot(m)$ space algorithm that outputs a collection of edge-disjoint cycles of length $O(\log^2n)$ that cover all but $O(n\log n)$ edges of the graph.
\end{lemma}
\begin{proof}
The algorithm in Section~3 of \cite{optimal_cycle_decomposition} first computes a \emph{low-congestion cycle cover}, which is a collection of short cycles that cover most of the edges such that each edge belongs to a small number of cycles. The precise parameters are specified in the statement of Lemma~2 of the full version of \cite{optimal_cycle_decomposition}: compute a collection of cycles of length $d=O(2^{1/\eps}\log n)$ that cover all but $O(n\log n)$ edges such that each edge appears on at most $c=1/\eps\cdot n^{O(\eps)}$ cycles. For $\eps=1/\log\log n$, we obtain $d=O(\log^2n)$ and $c=n^{o(1)}$. To obtain the short cycle decomposition, the proof of Lemma~2 describes a simple procedure of greedily selecting a subset of edge-disjoint cycles that cover an $\Omega(1/(dc))$ fraction of the edges in the cycle cover and then iterating on the uncovered edges.

Note that the low-congestion cycle cover consists of up to $cm$ edges since each edge can appear in $c$ cycles. Since $c=n^{o(1)}$, storing the entire cycle cover takes up too much space. Our key insight is to terminate the collection of cycles in each iteration of algorithm \textsc{ImprovedShortCycleDecomp} (see Figure 3 in the full version of \cite{optimal_cycle_decomposition}) once the cycle cover has a total of $O(m)$ edges. Terminating this subroutine early cannot increase the congestion compared to not terminating, thus we maintain the guarantee that the total set of cycles has congestion $c$. Using the above greedy selection of edge-disjoint cycles, we cover an $\Omega(1/(dc))$ fraction of the edges in our cycle cover, which is $\Omega(1/(dc))$ of all uncovered edges in the graph. Iterating on the uncovered edges gives the same iterations as before up to logarithmic factors. \qedhere
\end{proof}

With our improved short-cycle decomposition, we get the following lemma:
\begin{lemma}
    \label{lem:subtroutine}
    Given $\eps \in (0, 1]$, there is an algorithm we denote {\sc SpectralSketch}($G, \eps$), which given $G$, runs in time $m^{1 +o(1)}$ and space $\Ot(m)$, and with high probability returns a $(1 + \eps)$-for-each spectral sparsifier of $G$ with $\Ot(n/\eps)$ edges.
\end{lemma}

\subsection{Online Leverage Score Sampling}
\label{sec:online_sampling}
For the purposes of a faster runtime, we consider only sampling (and reweighting) a fraction of edges during the stream. Existing online leverage score sampling methods allow us to choose edges in an online stream without retracting our choices such that our final graph has $\Ot(n/\eps^2)$ edges and is a spectral-sparsifier of the original graph. Indeed, let $B_n \in \mathbb{R}^{\binom{n}{2}\times n}$ be the vertex edge incidence matrix of an undirected, unweighted complete graph on $n$ vertices, where the $e$-th row $b_e$ for edge $e = (u, v)$ has a $1$ in column $u$, a $(-1)$ in column $v$, and zeroes elsewhere. Then for an arbitrary undirected graph $G$, we can write its vertex edge incidence matrix $B = SB_n$ where $S \in \mathbb{R}^{\binom{n}{2} \times \binom{n}{2}}$ is a diagonal matrix with entry $\sqrt{w_e}$ in the $e$-th diagonal entry where $w_e$ is the weight of edge $e$. It is well-known that the Laplacian matrix $L$ of graph $G$ can be written as $L = B^\top B$. Hence, if we can sample a subset of the rows of $B$ to form a new matrix $C$ (which corresponds to a subset of the edges in $G$) such that for every $x$, $\|Bx\|_2 \approx_{\eps} \|Cx\|_{2}$, we then have $x^\top L x \approx_{\eps} x^\top C^\top C x$ so that the subgraph $H$ corresponding to $C$ is a $(1 + \eps)$-spectral sparsifier of $G$. For more details, we refer the reader to~\cite{KLM+17, CMP20}. Formally, we have the following lemma.

\begin{lemma}[Corollary 2.4 from \cite{CMP20}] \label{lem:online_sampling}
    Let $G$ be an undirected simple graph with $n$ vertices and poly$(n)$ bounded edge weights, and $\eps \in (0,1)$. We can construct a $(1+\eps)$-spectral sparsifier of $G$ as a reweighted subgraph with $O(n\log^2n/\eps^2)$ edges, using only $O(n\log^2n)$ bits of working memory in the online model. Additionally, the total running time is near-linear in the number of edges of $G$.
\end{lemma}

\subsection{Main Algorithm}
\label{sec:main_algorithm}
Our algorithm is presented in \cref{alg:dynamic_spectral_sparsifier}, which is based on the merge-and-reduce paradigm (see, for example, \cite{BDM+20}) and uses an additional factor of $\mathrm{polylog}(n)$ space. At a high level, our algorithm maintains a number of blocks of edges $\mathbf{B}_0, \mathbf{B}_1, \ldots \mathbf{B}_{\log (n/\eps)}$, each with size $\mathsf{m_{space}} = \Ot(n/\eps)$. The most recent edges are stored in $\mathbf{B}_0$; whenever $\mathbf{B}_0$ is full, the successive non-empty blocks $\mathbf{B}_0,  \dots, \mathbf{B}_i$ are merged and reduced to a new graph with $\mathsf{m_{space}}$ edges, which will be stored in $\mathbf{B}_{i+1}$.  
Next, we show the correctness of our algorithm. Following a similar argument in~\cite{BDM+20}, we first show the following lemma.

\begin{algorithm}
    \textbf{Input: }{Undirected graph $G(V,E)$ in a stream with $n$ vertices and $m$ edges, accuracy parameter $\eps \in (\frac{1}{n},1)$, offline spectral-sparsifier subroutine \textsc{SpectralSketch}}
    
    \textbf{Output: }{A graph $H$ with weights $w$.}
    
    Initialize $\mathsf{m_{space}} = n \log^c(n) / \eps$ for some constant $c$, $\mathbf{B}_i \gets \emptyset$
 
    \ForEach{sampled edge $e_t$ (\cref{lem:online_sampling}) with rescaled weight $u_t$}{
        \uIf{$\mathbf{B}_0$ does not contain $\mathsf{m_{space}}$ edges}{
        $\mathbf{B}_0\gets e_t \cup \mathbf{B}_0$\;
        }
        \Else{
            Let $i>0$ be the minimal index such that $\mathbf{B}_i=\emptyset$\;

            $\mathbf{B}_i, w_i \gets\textsc{SpectralSketch}\left(\textbf{M}, \frac{\eps}{\log (n/\eps)}\right)$, where $\textbf{M}=\mathbf{B}_0 \cup \dots \cup \mathbf{B}_{i-1}$\;

            \For{$j=0$ \KwTo $j=i-1$}{
                $\mathbf{B}_j\gets\emptyset$ \;
            }
        $\mathbf{B}_0 \gets e_t$\;
        }
    }
    $\mathbf{B} \gets \textsc{SpectralSketch}(\mathbf{B}_{\log (n/\eps)} \cup \cdots \cup \mathbf{B}_0, \eps)$ 
    
    \Return $\mathbf{B}$.
    \caption{\textsc{StreamingSpectralSparsifier($G, \eps$)}}
    \label{alg:dynamic_spectral_sparsifier}
\end{algorithm}

\begin{lemma}[see Lemma~5.2 in \cite{BDM+20}]
\label{lem:merge_and_reduce}
Suppose that $\mathbf{B}_0,\dots,\mathbf{B}_{i-1}$ are all empty while $\mathbf{B}_i$ is non-empty. Then $\mathbf{B}_i$ is a $(1 + \frac{\eps}{\log (n/\eps)})^i$-for-each spectral sparsifier for the last $2^{i - 1} \mathsf{m_{space}}$ edges. 
\end{lemma}

\begin{proof}
    We prove this by induction on $i\ge0$. Recall that $\mathbf{B}_i$ can only be non-empty if at some point $\mathbf{B}_0$ contains $\mathsf{m_{space}}$ and $\mathbf{B}_0, \mathbf{B}_1, \ldots, \mathbf{B}_{i - 1}$ are all non-empty. By induction, suppose that for every $1 \le j < i$, $\mathbf{B}_{j}$ is a $(1 + \frac{\eps}{\log(n/\eps)})^{j}$-for-each spectral sparsifier for $2^{j - 1}\mathsf{m_{space}}$ edges. Then we have that $\mathbf{B}_i$ is a $(1 + \frac{\eps}{\log(n/\eps)})$-for-each spectral sparsifier for $\mathbf{B}_0 \cup \mathbf{B}_1 \cup \ldots \cup \mathbf{B}_{i - 1}$. From the mergeability property of spectral sparsifier graphs, we get that $\mathbf{B}_i$ is a $(1 + \frac{\eps}{\log(n/\eps)}) (1 + \frac{\eps}{\log(n/\eps)})^{i - 1} = (1 + \frac{\eps}{\log(n/\eps)})^{i}$-for-each spectral sparsifier for the $\mathsf{m_{space}} + \sum_{j = 0}^{i - 2}2^j\mathsf{m_{space}} = 2^{i - 1}\mathsf{m_{space}}$ edges, as needed. 
\end{proof}

\begin{proof}[Proof of \cref{thm:sparsifier}]
    The success probability of each call to the subroutine $\textsc{SpectralSketch}$ is at least $1 - 1/\mathrm{poly}(n)$, hence we can assume that each call to this subroutine and the online leverage score sampling procedure is successful with high probability after taking a union bound. It then follows from \cref{lem:merge_and_reduce} that $\mathbf{B}_{\log (n/\eps)} \cup \cdots \cup \mathbf{B}_0$ is a $(1 + \frac{\eps}{\log (n/\eps)})^{\log (n/\eps)} \leq (1+\eps)$-for-each spectral sparsifier of the original graph $G$. Thus $\mathbf{B}$ is a $(1 + O(\eps))(1 +\eps) = (1 + O(\eps))$-for-each spectral sparsifier of the original graph $G$.

    \paragraph{Space Complexity.} 
    Next, we analyze the space complexity of our algorithm. As stated in \cref{lem:online_sampling}, the online leverage score sampling procedure can be implemented in $O(n\log^2 n)$ bits of working memory. We maintain $\log (n/\eps)$ blocks $\mathbf{B}_i$, each taking at most $O(\mathsf{m_{space}})$ words of space. From \cref{lem:spectral_sketch} we have that each call to the subroutine $\textsc{SpectralSketch}$ takes at most $\Ot(\mathsf{m_{space}} \cdot \log (n/\eps))$ words of space, as the total number of edges never exceeds $\mathsf{m_{space}} \cdot \log (n/\eps)$. Hence, the total space of the algorithm is 
    $\Ot(n \log^2 n + \mathsf{m_{space}} \log (n/\eps)) = \Ot(n/\eps)$.

    \paragraph{Time Complexity.} 
    We finally analyze the time complexity of our algorithm. As stated in \cref{lem:online_sampling}, the online leverage score sampling procedure can be done in time $\Ot(m)$, and after the sampling procedure, there are at most $O(n \log^2 n /\eps^2)$ edges. For the merge-and-reduce procedure, each edge participates in at most $\log (n/\eps)$ different spectral sparsifiers, one in block $\mathbf{B}_i$. Additionally, from \cref{lem:spectral_sketch} we have that $\textsc{SpectralSketch}$ runs in $m^{1+o(1)}$ time where $m = \Ot(n/\eps)$. We deduce that our total runtime is 
    $\Ot(m + (n\log^2 n/\eps^2)^{1 + o(1)}) = \Ot(m) + (n/ \eps^2)^{1 + o(1)}$.
\end{proof}

\subsection{Approximate Minimum Cut Streaming Algorithm}
\label{sec:min_cut}
Our main application of \cref{thm:sparsifier} is the one-pass streaming algorithm for a $(1 + \eps)$-approximation to the minimum cut. To achieve this,
first recall that we can get a $(1 + \eps)$ for-all sparsifier of $G$ in a one-pass stream in nearly-linear time, using $\widetilde{O}(n/\eps^2)$ bits of space (\cref{lem:for_all}).

The next lemma bounds the number of approximate minimum cuts.
\begin{lemma}[\cite{karger00}]
\label{lem:cut_number_constant}
For constant $\alpha > 0$, the number of $\alpha$-approximate minimum cuts is $O(n^{\lfloor 2 \alpha \rfloor})$.    
\end{lemma}

Motivated by \cite{AndoniCKQWZ16}, we run our algorithm in \cref{thm:sparsifier} along with the streaming for-all sparsifier algorithm $\mathsf{ALG} $ in \cref{lem:for_all} with accuracy parameter $\eps' = O(1)$. Suppose that the output of $\mathsf{ALG}$ is $H$. We then use $H$ to find the cuts that have a size less than $1.5$ times that of the minimum cut value (which can be done in $\mathrm{poly}(n)$ time, see, e.g., \cite{karger00}). Then, since the number of these approximate minimum cuts is $O(n^3)$, using our $(1 + \eps)$ for-each sparsifier, we can get a $(1 + \eps)$-approximation to each of these cut values with high probability after taking a union bound. Picking the cut with the minimum value gives us a $(1 + \eps)$-approximate minimum cut. 

\paragraph{Faster Post-processing Time.} 
The above algorithm is simple but, unfortunately, takes $O(n^3)$ post-processing time. To achieve a faster post-processing time, we consider the recursive contraction algorithm introduced by~\cite{KS96}, which enumerates all $\alpha$-approximate minimum cuts in time $O(n^{2\alpha} \log^2 n)$. We combine this with a Johnson-Lindenstrauss sketch to achieve $\Ot(n^2/\eps^2)$ post-processing time, proving \cref{corollary:min_cut}. 

First, we state the lemma from \cite{KS96}:
\begin{lemma}
    \label{lem:rec_alg}
    There is an algorithm that can find all the $\alpha$-approximate minimum cuts with high probability in $O(n^{2\alpha} \log^2 n)$ time and $\Ot(m)$ space.
\end{lemma}

Below, we first give a brief explanation of the algorithm. At a high level, in each step, the algorithm randomly samples an edge proportional to its weights and then contracts the two nodes of the edge. When there are only two remaining nodes, we get a partition of the nodes corresponding to a specific cut of the graph. The work of~\cite{KS96} shows that if we repeat the algorithm $\mathrm{poly}(n)$ times, we can finally enumerate all approximate minimum cuts with high probability. The above process is simple but needs $\mathrm{poly}(n)$ runtime. To make the algorithm faster,~\cite{KS96} then proposes a recursive implementation of the contraction process. Basically, at each level of the recursion, the algorithm reduces the number of nodes by contraction by a constant factor and repeats the algorithm twice in each recursion level if the number of remaining nodes is $\Omega(1)$. Then, each leaf node of the recursion tree corresponds to a specific cut of the graph.

The technical difficulty in our case is evaluating the value of the corresponding cut in a short time when we enter a leaf node in the above process. To achieve this, recall that for a given vertex set $S \subseteq V$, the cut value between $S$ and $V \setminus S$ is $x_S^\top L x_S = x_S^\top B^T B x_S = \|Bx_S\|_2^2$ where $x_S$ is the binary indicator vector of $S$ and $B$ is the edge-vertex matrix of the graph (see \cref{sec:online_sampling} for more details). The key observation here is since we only care about a $(1 + \eps)$-approximate value, we can apply a JL matrix $T$ to $Bx_S$, where the matrix $T$ has only $O(\log (n/\eps) / \eps^2)$ rows where for every $x \in \R^{O(n/\eps)}$, $\|Tx\|_2^2 \approx_\eps \|x\|_2^2$ with high probability.

Our algorithm procedure is described as follows. As before, after the stream, we first get a for-all sparsifier $K$ with the accuracy parameter $\eps' = 1/\log n$ (\cref{lem:for_all}) and a $(1 + \eps)$ for-each sparsifier $H$ of the original graph $G$ (\cref{thm:sparsifier}). We then use the algorithm in \cref{lem:rec_alg} to enumerate all the approximate minimum cuts of $K$ with $\alpha = 1 + c/\log n$. When maintaining the recursion process of the contraction algorithm, we also maintain a sketch of the columns of the edge-vertex matrix $B$ of $H$. Specifically, we initially compute a sketch of $TB$ where $T$ is a JL matrix with $O(\log n/\eps^2)$ rows. Then, in the contraction process, when we contract a node pair $(u, v)$, we also merge the $u, v$-th columns of the sketch with their sums. Note that in this procedure, when we enter a leaf node, there will only remain two columns of our sketch, and the two vectors are exactly $TBx_S$ and $TBx_{V \setminus S}$, for which we can directly compute the cut value in $O(\log n / \eps^2)$ time.

Next, we analyze the time and space complexity in the above procedure. Since the sketch has $O(\log n / \eps^2)$ rows at each recursion level, we need $\Ot(N\log n / \eps^2)$ words of space to save the sketch, where $N$ is the number of the remaining nodes in the current recursion level. Combining two rows also takes $\Ot(\log n / \eps^2)$ time. Recall that we need $O(n)$ space in each level and $O(N)$ time to contract two nodes in the original recursive contraction algorithm. This implies after the modification, the space complexity of the above algorithm procedure remains the same of $\Ot(m) = \Ot(n/\eps)$, and the time complexity becomes $\Ot(n^2 / \eps^2)$. Putting everything together, we get the correctness of our \cref{corollary:min_cut}.

\subsection{Approximate Minimum Cut Streaming Lower Bound} 
To prove this lower bound, we consider the $k$-edge-connectivity problem as follows. The \emph{edge-connectivity} of a graph is the minimum number of edges that need to be deleted to disconnect the graph. The $k$-edge-connectivity problem asks whether the edge connectivity of a graph is $<k$ or $\geq k$.
The work of \cite{SW15} shows that any deterministic algorithm that solves the $k$-edge-connectivity problem in insertion-only streams needs $\Omega(kn \log n)$ bits of space. We note that the proof in~\cite{SW15} also works for randomized algorithms with success probability at least $1 - 1/n$, which implies a lower bound of $\Omega(kn)$ bits of space for any randomized algorithm with constant probability of success.
\begin{claim}[\cite{SW15}]\label{prop:edge-connect}
    Any randomized algorithm solving the $k$-edge-connectivity problem in an insertion-only stream with probability at least $2/3$ requires $\Omega(kn)$ bits of space. Moreover, any deterministic algorithm solving the $k$-edge-connectivity problem in an insertion-only stream requires $\Omega(kn \log n)$ bits.
\end{claim}

We now prove the randomized and deterministic lower bounds of \cref{thm:lower_bound}.

\paragraph{Randomized Lower Bound.}
\begin{lemma}\label{lemma:randomized_lower}
    Fix $\eps > 1/n$. Any randomized algorithm that outputs a (1 + $\eps$)-approximation to the minimum cut of a simple, undirected graph in a single pass over a stream with probability at least 2/3 requires $\Omega(n/\eps)$ bits of space
\end{lemma}
\begin{proof}
Consider the lower bound for $k$-edge-connectivity in \cite{SW15} with $k= \frac{1}{10 \eps}$ (\Cref{prop:edge-connect}). Distinguishing between whether the graph is $<k$-connected or $\geq k$-connected needs $\Omega(kn)=\Omega(n/\eps)$ space. This implies that deciding whether the min-cut has value $<k$ or $\geq k$ also requires $\Omega(kn)=\Omega(n/\eps)$ space. Suppose that there is a $(1+\eps)$-approximation algorithm for the minimum cut value, in the case where the edge connectivity $\lambda < k$, the approximate value is at most $(1+\eps) \lambda \leq (1+1/10k) \cdot (k-1) < k - 9/10$. In the case where the connectivity $\lambda \geq k$, the approximate value is at least $(1-\eps) \lambda \geq (1- 1/10k) \cdot k > k - 1/10$. This implies that an algorithm that gives a $(1+\eps)$-approximation to minimum cut value, solves $k$-edge-connectivity for $k= \frac{1}{10\eps}$. Thus, we obtain a lower bound of $\Omega(n/\eps)$ bits of space for $(1+\eps)$-approximating the value of the minimum cut.
\end{proof}
In addition to the proof above that considers the $k$-edge-connectivity problem, we also provide an alternative self-contained construction for the randomized result using the
Index communication problem in \Cref{sec:lb}.

\paragraph{Deterministic Lower Bound.}
If the minimum cut value algorithm is deterministic, then we deterministically solve $k$-edge-connectivity for $k= \frac{1}{10\eps}$. This gives us a lower bound of $\Omega(kn \log n) = \Omega(n \log n/\eps)$ bits of space for deterministic algorithms.
We show how to improve this lower bound to $\Omega(n/\eps^2)$ bits for $\eps \geq n^{-1/4}$.
Note that this assumption on $\eps$ is consistent with the previous work of \cite{carlson2019optimal}.
They also use this assumption on $\eps$ and prove a lower bound of $\Omega(n \log n/\eps^2)$ bits for cut sparsifiers.

\begin{lemma}\label{lemma:deterministic_lower}
    Fix $\eps > 1/n^{1/4}$. Any deterministic algorithm that outputs a (1 + $\eps$)-approximation to the minimum cut value of a simple, undirected graph in a single pass over a stream requires $\Omega(n/\eps^2)$ bits of space
\end{lemma}
\begin{proof}
A randomized insertion-only streaming algorithm for a for-all cut sparsifier needs $\Omega(n/\eps^2)$ bits of space \cite{AndoniCKQWZ16} (since any data structure that stores this information needs that much space). If we carefully look at the proof of \cite{AndoniCKQWZ16} in section $3.1$, their hard instance is a disjoint union of $\eps^2 n/2$ bipartite graphs with $2/\eps^2$ vertices each. They use the cut sparsifier to query cuts that have vertices only in one of these disjoint graphs, which implies that the cut sparsifier only needs to preserve the cut values for cuts that have at most $1/\eps^4$ edges. Thus, a $(1+\eps)$ cut sparsifier that preserves the cut values of all cuts of size at most $1/\eps^4$ edges needs $\Omega(n/\eps^2)$ bits of space.

We will now show how to simulate a $(1+\eps)$ cut sparsifier that preserves the cut values of all cuts of size at most $1/\eps^4$ using an algorithm for $(1+\eps)$-approximate minimum cut. This will give us the desired lower bound.
Let $A$ be a deterministic insertion-only streaming algorithm for $(1+\eps)$-approximate minimum cut value we run during the stream for input graph $G$.

After the stream, consider any cut $S$ whose value we want to approximate within a $(1+\eps)$ factor (note that the value is at most $1/\eps^4$). We now add some extra edges to the stream for $A$. We first construct two cliques $C_1$ and $C_2$ on $2n$ vertices each and add their edges to the stream. We then add edges between all vertices of $S$ and $C_1$ and similarly for all vertices in $\bar{S}$ and $C_2$. Call this new graph $G'$.
We now end the stream and look at the output of $A$. We claim this is the $(1+\eps)$-approximate value of cut $S$ in the original graph.

To prove this, we must show that $S$ is the minimum cut in the modified graph $G'$. We know that the value of $S$ is at most $1/\eps^4 \leq n$ (since $\eps \geq n^{-1/4}$).
Consider any cut that separates vertices in $C_1 \cup S$. Such a cut has a size of at least $2n$. The same applies to any cut separating vertices in $C_2 \cup \bar{S}$. Thus, the minimum cut is $C_1 \cup S$. None of the new edges we added cross this cut, so the size of this cut remains unchanged in $G'$, implying that we get a $(1+\eps)$-approximation to the cut value in $G$. Note that the algorithm $A$ has to be deterministic for us to do this because we have to repeat this process for exponentially many cuts (potentially $2^n$).

The space taken to build the sparsifier for insertion-only streams is the space taken for $A$ when run on $5n$ vertices.
This implies that any deterministic insertion-only streaming algorithm for $(1+\eps)$-approximate minimum cut value needs $\Omega(n/\eps^2)$ bits of space, proving the lower bound for deterministic algorithms.
\end{proof}

\begin{proof}[Proof of \cref{thm:lower_bound}]
The proof of the theorem directly results from combining \cref{lemma:randomized_lower} and \cref{lemma:deterministic_lower}. Notably, this shows that the spectral sparsifier result in \cref{thm:sparsifier} is optimal in space complexity up to polylogarithmic factors. 
\end{proof}

\subsection{Application: All-Pairs Effective Resistances}
\label{sec:resistance}
Another application of our streaming algorithm from \cref{thm:sparsifier} is calculating all-pairs effective resistances in graph streams. Suppose we have a $(1 + \eps)$ for-each spectral sparsifier $H$ of the original graph $G$. Then, we may hope it will be helpful to approximate the pseudoinverse of $G$. In fact, in the work of~\cite{jambulapati2018efficient}, it has been shown that if we have a $(1 + \eps)$ for-each spectral sparsifier, we can use it to generate an effective resistance sketch in near-linear time, which can compute the all-pairs effective resistances in $\widetilde{O}(n^2 / \eps)$ additional time (Algorithm~8 in~\cite{jambulapati2018efficient}). From this, we immediately see the correctness of our \cref{corollary:efffective_resistance}.

\paragraph{Lower Bound.}
We next consider lower bounds. In \cref{thm:lower_bound_effective_resistance}, we show that if a sketch can approximate each effective resistance with constant probability, then such a sketch must have size at least $\Omega(n/\eps)$. We need the following lemmas about effective resistances and random graphs.

\begin{lemma}(see, e.g., \cite{Lec13})
    \label{lem:commute time}
    For a graph $G$ with unit weight for each edge, we have $C(s, t) = 2m r_{\mathrm{eff}}(s,t)$, where $C(s, t)$ is the commute time of $s$ and $t$, which is defined as the expected number of steps of a random walk to go from $s$ to $t$ and back.
\end{lemma}

\begin{lemma}[\cite{ottolini2023concentration}]
\label{lem:random_graph}
    Let $G(n, p)$ be an Erdos-Renyi random graph. Then we have, with high probability that for any pair $w \ne v$,
    \begin{equation*}
        h(w, v) = \frac{2|E|}{\mathrm{deg}(v)} + \left\{\begin{aligned}
        &-1 \ & \text{if $(w, v) \in E$}\\
        &-1 + 1/p \ & \text{if $(w, v) \notin E$} 
    \end{aligned}\right\} + O\left(\frac{(\log n)^{3/2}}{\sqrt{n}}\right),
    \end{equation*}
    where $h(s, t)$ is the hitting time of $s, t$, defined as the expected number of steps of a random walk from $s$ to $t$. Note that we have $C(s, t) = h(s, t) + h(t, s)$.
\end{lemma}

Our reduction is based on the following Index problem.

\begin{lemma}[Index Lower Bound \cite{kushilevitz_nisan_1996}]
\label{lem:index}
Suppose that Alice has a random string $u \in \{0, 1\}^n$ and Bob has a random index $i \in [n]$. If Alice sends a single message to Bob from which Bob can recover $u_{i}$ with probability at least $2/3$, where the probability is over both the randomness of the protocol and the input, then Alice must send $\Omega(n)$ bits to Bob.
\end{lemma}

At a high level, we will reduce the Index problem to the effective resistance problem. Suppose Alice has a string $s \in \{0, 1\}^{\Theta(n/\eps)}$. We will construct a graph $G$ to encode $s$ such that Bob can recover $s_i$ with constant probability from a $(1\pm\eps)$ effective resistance sketch of $G$. By the communication complexity lower bound of the Index problem (\cref{lem:index}), the cut sketch must use  $\Omega(n/\eps)$ bits. 

\begin{proof}[Proof of \cref{thm:lower_bound_effective_resistance}] We first give our construction of the graph $G$.

\paragraph{Construction of $G$.}
We use a bipartite graph $G$ to encode $s$. Let $L$ and $R$ be the left and right nodes of $G$ where $|L| = |R| = n/2$. We partition $L$ into $O(n\eps)$ disjoint blocks $L_1, \ldots, L_{O(n\eps)}$ of the same size $\frac{1}{\eps}$, and similarly, we partition $R$ into $R_1, \ldots, R_{O(n\eps)}$. We divide $s$ into $n\eps$ disjoint strings $s_{i} \in \{0, 1\}^{(\frac{1}{\eps})^2}$ of the same length. We will encode $s_{i}$ using the edges from $L_i$ to $R_i$, where we use $G_i$ to denote the subgraph $(L_i, R_i)$. In particular, if the $s_{i}(j) = 1$, we form an edge that connects the corresponding node pair in $L_i$ and $R_i$ with unit weight. Note that each $G_i$ is disconnected from the other $G_j$ in this construction.

\paragraph{Recovering a bit in $s$ from an  effective resistance sketch of $G$.}
Suppose Bob wants to recover a specific bit of $s$, which belongs to the sub-string $z = s_{i}$ and has an index $t$ in $z$. This coordinate corresponds to whether there is an edge between $u \in L_i$ and $v \in R_i$.

Next, we consider the effective resistance between $u$ and $v$ for these two cases. Note that since $G_i$ is disconnected from the other $G_j$, we just need to consider the sub-graph $G_i$. By construction, $G_i$ is a random graph with $p = \frac{1}{2}$. From~\cref{lem:commute time} we get that we only need to show both the hitting time $h(u, v)$ and $h(v, u)$ has a $(1 + \eps)$-separation (then the $r_{\mathrm{eff}}(u, v)$) will also have a $(1 + \eps)$-separation. Then, from~\cref{lem:random_graph}, we get that both the $h(u, v)$ and $h(v, u)$ will have a $\Theta(1)$ gap for the two cases, and with high probability, we have that both of  $h(u, v)$ and $h(v, u)$ are $\Theta(1/\eps)$. From this, we can get that $r_{\mathrm{eff}}(u, v)$ will have a $(1 + \eps)$-separation for the two cases (note that Alice can also send the degree of each node and the number of edges in each $G_i$ to Bob, which only needs $O(n\log(1/\eps))$ bits), which yields an $\Omega(n / \eps)$ lower bound of the effective resistance sketch size.
\end{proof}
\section{Exact Minimum Cut in Random-Order Streams} \label{sec:random_order}

In this section, we give an $\widetilde{O}(n)$ algorithm that outputs the minimum cut in a simple, unweighted graph in a single-pass random-order stream. Recall that in a simple graph, there is at most one edge between any pair of vertices. For a clearer demonstration, we will present an initial construction for an algorithm with an $\widetilde{O}(n^2)$ update time when an edge arrives. Then, we shall show how to improve the update time to $\widetilde{O}(n)$, proving \cref{thm:random_order}.

\subsection{Initial Construction}
We will use the algorithm in \cref{lem:for_all} (for-all sparsifier in a stream) as a subroutine. The whole algorithm is given in \cref{alg:random_order}. Below, we first give a high-level explanation of our algorithm. We first consider the case when the minimum cut size is $s = \Omega(\log n)$. Since we are now considering the random-order model, a prefix of the edges gives us partial information about the cut sizes. Let $H$ be the subgraph of $G$ which is formed by a prefix of the edges of the stream where $|H| \approx \frac{|G| \log n}{s}$ and $|G|$, $|H|$ are the number of the edges in $G$ and $H$ respectively. By a Chernoff bound, one can show that with high probability, for every subset $S \subset V$, $w_H(S, V\setminus S)$ is a small constant approximation to $w_G(S, V\setminus S)$. 
$H$ might be too large to store in memory, so we store $H_1$, a $(1 + \eps)$-for-all sparsifier of $H$, during the stream (we will decide $\eps$ later).
Therefore, using the graph $H_1$, we know all cuts in $G$ whose cut size is within a factor of $1.1$ of the true minimum cut of $G$. After this step, when a new edge $e \in G \setminus H$ arrives, we can use graph $H_1$ to check whether $e$ belongs to some $1.1$-approximate minimum cut in $G$, and if so, we save this edge (from \cref{lem:edge_number} we know that the total number of these edges saved is at most $O(n)$). The next question is how to estimate the exact value of these $1.1$-approximate cuts in $H$. The crucial observation is since the minimum cut size of $H$ is an integer and is at most $\Theta(\log n)$, we can set the approximation parameter $\eps = 1/\log^2 n$ for $H_1$. This gives the exact value of $w_H(S, V \setminus S)$ if $S$ and $V \setminus S$ is a constant-approximate minimum cut in $G$ (since the error is at most $O(\log n) \cdot 1/\log^2 n < 1$). Putting these things together, after the stream, we can enumerate all $1.1$-approximate minimum cuts and obtain their exact values, thus getting the exact minimum cut value. 

The remaining case to handle is when the minimum cut size is $O(\log n)$. Similarly, in this case, setting $\eps = 1/\log^2 n$ will give us the exact value of the minimum cut.
Finally, we do not know the value of $s$, so we will try all powers of $2$ for it.
Note that we do not need the exact value; a $2$-approximation to $s$ will suffice.
The detailed description of the algorithm is as follows:
\begin{algorithm}[ht]
    \SetKwInOut{Input}{Input}
    \Input{Undirected and unweighted graph $G(V,E)$ in a random-order stream with $n$ vertices and $m$ edges.}
    $\mathsf{ALG_1}$ is an instance of \cref{lem:for_all} with $\eps = 1/\log^2 n$.
    
    Maintain the degree of each node $d_i$ during the stream.
    \caption{\textsc{MinCutRandomOrder}}
    \label{alg:random_order}
    \ForEach{edge $e$ in the graph stream} {
        Feed $e$ to $\mathsf{ALG_1}$ \; 
        \If{$e$ is the $2^i$-th edge for some $i \in \mathbb{N}$}{
        $H_1 \gets$ output of $\mathsf{ALG_1}$ (Here $H_1$ is a for-all sparsifier of the prefix graph $H$)\; 
            \If{the minimum cut size of $H_1$ is larger than $c \log n$}{
                Save $H_1$ and break the for loop.
            }
        }
    }
    $H_1 \gets$ output of $\mathsf{ALG_1}$ \; 
    \If{there is no new edge in the stream}{
    \Return{the minimum cut value of $H_1$}.
    }
    $T \gets \emptyset$.
    
    \ForEach{new edge $e$ during the stream} {
        \If{$e$ is the cut edge between some $S$ and $V \setminus S$ where $S$ and $V \setminus S$ is a non-singleton $1.1$-approximate minimum cut in $H_1$}{
        Add $e$ to $T$
        }
    }
    \ForEach{$S$ where $S$ and $V \setminus S$ is a non-singleton $1.1$-approximate minimum cut in $H_1$}{$v_S \gets w_{H_1}(S, V \setminus S) + |\{e \in T: e \text{ is a cut edge between $S$ and $V \setminus S$}\}|$ }
    \Return{the minimum value of $\min v_S$ and  $\min d_i$}
\end{algorithm}

To show the correctness of our algorithm, we first prove the following key lemma.

\begin{lemma}
\label{lem:cut_concentration}
    Suppose that $H$ is a subgraph of $G$ formed by a prefix of the edges in the random order model. We have that for every $S \subset V$, if $|H| = |G| \cdot \frac{\ell}{ w_G(S, V \setminus S)} $, then with probability at least $1 - 2e^{-c_2 \cdot \ell}$, $|w_H(S, V \setminus S) - \ell | \le 0.1 \ell $.
\end{lemma}

\begin{proof}
    Consider each edge $e$ that is a cut edge between $S$ and $V \setminus S$ (let $T$ denote the set of such edges). Let $y_e$ denote the indicator random variable where $y_e = 1$ if $e \in H$ and $y_e = 0$ otherwise. Then we have that $\mathbb{E}[w_H(S, V \setminus S)] = \mathbb{E}[\sum_{e \in T} y_e] = \ell$ from the assumption in the random-order model. By a standard Chernoff bound, we have that 
    \[
    \Pr[|w_H(S, V \setminus S) - \ell |\ge 0.1 \ell] \le 2e^{-c_2 \cdot \ell}
    \]
    for some constant $c_2$.
\end{proof}

We also need the following structural result for the approximate minimum cuts of a graph.

\begin{lemma}[\cite{rsw18}]
    \label{lem:edge_number}
    For any simple, unweighted graph and any constant $\eps>0$, the total number of edges that participate in non-singleton $(2-\eps)$-approximate minimum cuts is at most $O(n)$.
\end{lemma}

We first consider the case when the minimum cut size $s < c \log n$. In this case, the minimum cut of $H$ will never exceed $c \log n$ since $H$ is a subgraph of $G$. Since the cut value is an integer, setting $\eps = 1/\log^2 n$ gives the exact value of the size of the minimum cut.

We next consider the other case when $s \ge c \log n$. We will first show that, as the number of edges in $H$ increases, the minimum cut size of $H$ will increase to $\Theta(\log n)$ when $|H| = \Theta\big(\frac{|G| \log n}{s}\big)$. We next analyze the value of $w_H(S, V \setminus S)$. Let $\mathcal{S}_i$ ($1 \le i \le \log n$) denote the set of nodes such that
\[
\mathcal{S}_i = \{S \subset V \ | \ s\cdot 2^{i - 1} \le w_G(S, V \setminus S) \le s \cdot 2^i \} \;
\]
For each $\mathcal{S}_i$, from \cref{lem:cut_number_constant} we know that $|\mathcal{S}_i| \le O(n^{2^{i + 1}})$ and for every $S \in \mathcal{S}_i$, from \cref{lem:cut_concentration} we know that with probability at least $1 - 2e^{-c_2 c \cdot \log n \cdot 2^{i - 1}} = 1 - 2n^{-c_2 c \cdot 2^{i - 1}}$, $w_H(S, V \setminus S)\cdot \frac{|G|}{|H|}$ is a $1.1$-approximation of $w_G(S, V \setminus S)$. Note that the constant $c$ here can be sufficiently large, and after taking a union bound over all $S \in \mathcal{S}_i$ and all $\log n$ $\mathcal{S}_i$ we get that with probability at least $9/10$, for every $S \subset V$, $w_H(S, V \setminus S)\cdot \frac{|G|}{|H|}$ is a $1.1$-approximation of $w_G(S, V \setminus S)$. From the above analysis we immediately know that when $\frac{|G|}{|H|} = \Theta\big(\frac{s}{\log n}\big)$ from the subgraph $H$, we can learn the set 
\[\mathcal{S} = \{S \subset V \ | \ s \le w_G(S, V \setminus S) \le 1.1^2 s \},\]
and the next step is choosing the true minimum cut among them. For every non-single node set $S \in \mathcal{S}$, as mentioned, since $w_H(S, V \setminus S)$ is $\Theta(\log n)$, setting $\eps = 1/\log^2 n$ we can get its exact value. The remaining step is to estimate the value of $w_{G \setminus H}(S, V \setminus S)$. As described in \cref{alg:random_order}, after the previous step, we save every edge $e \in G \setminus H$, where $e$ belongs to at least one of the non-singleton $1.1^2$-approximate minimum cuts in $\mathcal{S}$. We know the exact value of $w_{G \setminus H}(S, V \setminus S)$ from this edge set. Putting the two things together, for every $S \in \mathcal{S}$ and $|S| \ge 2$, we know the value of $w_G(S, V \setminus S)$ after taking a sum of the two parts. For every singleton cut, we simply maintain the degree of each node during the stream. Thus, after taking the minimum value of the two parts (singleton and non-singleton cuts), we get the exact value of the minimum cut.

\paragraph{Space Complexity.} In the first phase of the algorithm, we use one instance of the algorithm in \cref{lem:for_all} with $\eps = 1/ \log^2 n $, which takes space $\Ot (n)$. In the second phase of the algorithm, we save all the edges that belong to the approximate minimum cut in $H$. From \cref{lem:edge_number}, we know that there are at most $O(n)$ such edges, and hence this part takes $O(n)$ words of space. We also maintain the degree of each node during the stream, which takes $O(n)$ words of space. Putting everything together, the space usage of our algorithm is $\widetilde{O}(n)$.

\paragraph{Time Complexity.} In the first phase of the algorithm, when one edge $e$ comes, the update and query time of $\mathsf{ALG}_1$ is polynomial, and we can also find the minimum cut in polynomial time. Hence, the update time here is still polynomial. In the second phase, when one edge $e$ comes, we can enumerate all $1.1$-approximate minimum cuts to check whether $e$ belongs to one of the approximate minimum cuts (we can use $O(n^2)$ time to enumerate all approximate minimum cuts, see, e.g.,~\cite{karger00}). Hence, we can do the update of this step in polynomial time. After all edges come, we can similarly enumerate all minimum cuts, from which the overall algorithm can be implemented in polynomial time.

\subsection{Faster Update and Post-Processing Time} In the above algorithm, the runtime bottleneck is when a non-prefix edge $e$ comes; we need to enumerate all approximate minimum cuts in the prefix graph to check whether to keep this edge. To get a faster runtime, we instead do a check when we collect $n$ edges and consider a similar procedure to what we did in \cref{sec:min_cut} where we use the recursive contraction algorithm with parameter $\alpha = 1 + \frac{1}{\log n}$. The difference is when doing the recursive contraction algorithm, we maintain the sketch $SB$ where $S$ is a $k$-sparse recovery matrix (see, e.g., \cite{GLPS12}) with $k \log (n / k)$ rows where $k = O(\log n)$ and $B$ is the edge-vertex matrix of the $n$ edges we currently collected. Particularly, during the recursive contraction process, when we contract the nodes $u, v$, we merge the columns that $u, v$ corresponds to in the sketch $SB$ and replace them with their sum. Next, we consider each leaf node in the recursion, corresponding to one specific cut for the prefix graph. We can check whether it is a non-singleton and approximate minimum cut in $O(1)$ time from the information the algorithm keeps. Note that if without the sketch matrix $S$, for each edge $e$ we want to check, it belongs to this specific cut if and only if the $e$-th coordinates of the remaining two columns of the edge-vertex matrix are $1$ and $-1$, and otherwise these two coordinates are both $0$. Since the minimum cut of the prefix graph is $\Theta(\log n)$, we can get that both of the remaining columns are $O(\log n)$-sparse before multiplying the sketch matrix $S$. This means that a $k$-sparse approximate recovery algorithm with $k = O(\log n)$ is sufficient to recover the indices of the non-zero coordinates of the remaining columns, which helps us to find the corresponding edges.   

\paragraph{Time and Space Complexity}
We analyze the time and space complexity of the above procedure. At each level of the recursion, since the sketch has $k \log(n / k)$ rows, we need $\Ot(Nk)$ words of space to save the sketch, and it takes $\Ot(k)$ time to combine two rows, where $N$ is the number of the remaining nodes in the current level of the recursion. Recall that we need $O(n)$ space in each level and $O(N)$ time to contract two nodes in the original recursive contraction algorithm. This implies the modification only increases the time and space complexity by a factor of $\Ot(k)$. Since the decoding time of the $k$-sparse recovery sketch is $k \cdot \mathrm{polylog} (n)$ and $k = O(\log n)$, we have that the procedure has time complexity $\Ot(n^2)$ and space complexity $\Ot(n)$. Finally, note that when we collect a set of $n$ edges in the suffix graph, we can do our above checking procedure while collecting the following $n$ edges in the graph stream, which results in a strictly $O(n)$ update time of our algorithm. Similarly, suppose the minimum cut of the graph is $\Theta(c)$. When we do the post-processing, we can use the $c$-sparse recovery sketch when enumerating all the approximate minimum cuts, which results in a $\Ot(n^2 c)$ post-processing time. Putting everything together, we have the correctness of \cref{thm:random_order}.
\vspace{0.5cm}

Lastly, we note that our algorithm can not only return the exact min-cut value but also collect all the edges that participate in any of the minimum cuts. When computing the for-all sparsifier of the prefix graph, for an edge that participates in at least one approximate min-cut, the probability that it will be sampled is $\Omega(1)$ as otherwise, the estimated cut value in such a cut will have two different values each with a constant probability, contradicting the guarantees of the for-all sparsifier. Hence, we can collect all of these edges from oversampling by a constant factor, which results in finding all the minimum cuts of the input graph and the edges crossing them.

\section*{Acknowledgments}
The authors would like to thank the ITCS 2025 reviewers for their anonymous feedback. Alexandro Garces and Vihan Shah are extremely grateful to Sepehr Assadi for many helpful conversations throughout the project. They also thank the organizers of DIMACS REU in Summer 2023, in particular Lazaros Gallos, for initiating this collaboration and for all their help and encouragement along the way.

\bibliographystyle{alpha} 
\bibliography{citation}

\newcommand{\etalchar}[1]{$^{#1}$}
\begin{thebibliography}{KMM{\etalchar{+}}20}

\bibitem[AB21]{assadi_et_al:LIPIcs.ICALP.2021.19}
Sepehr Assadi and Soheil Behnezhad.
\newblock {Beating Two-Thirds For Random-Order Streaming Matching}.
\newblock In Nikhil Bansal, Emanuela Merelli, and James Worrell, editors, {\em 48th International Colloquium on Automata, Languages, and Programming (ICALP 2021)}, volume 198 of {\em Leibniz International Proceedings in Informatics (LIPIcs)}, pages 19:1--19:13, Dagstuhl, Germany, 2021. Schloss Dagstuhl -- Leibniz-Zentrum f{\"u}r Informatik.

\bibitem[ABB{\etalchar{+}}19]{assadi2019coresets}
Sepehr Assadi, MohammadHossein Bateni, Aaron Bernstein, Vahab Mirrokni, and Cliff Stein.
\newblock Coresets meet edcs: algorithms for matching and vertex cover on massive graphs.
\newblock In {\em Proceedings of the Thirtieth Annual ACM-SIAM Symposium on Discrete Algorithms}, pages 1616--1635. SIAM, 2019.

\bibitem[ACK{\etalchar{+}}16]{AndoniCKQWZ16}
Alexandr Andoni, Jiecao Chen, Robert Krauthgamer, Bo~Qin, David~P. Woodruff, and Qin Zhang.
\newblock On sketching quadratic forms.
\newblock In Madhu Sudan, editor, {\em Proceedings of the 2016 {ACM} Conference on Innovations in Theoretical Computer Science (ITCS)}, pages 311--319, 2016.

\bibitem[AD21]{two_pass}
Sepehr Assadi and Aditi Dudeja.
\newblock {\em A Simple Semi-Streaming Algorithm for Global Minimum Cuts}, pages 172--180.
\newblock Society for Industrial and Applied Mathematics, 01 2021.

\bibitem[Alo97]{alon1997edge}
Noga Alon.
\newblock On the edge-expansion of graphs.
\newblock {\em Combinatorics, Probability and Computing}, 6(2):145--152, 1997.

\bibitem[BDM{\etalchar{+}}20]{BDM+20}
Vladimir Braverman, Petros Drineas, Cameron Musco, Christopher Musco, Jalaj Upadhyay, David~P. Woodruff, and Samson Zhou.
\newblock Near optimal linear algebra in the online and sliding window models.
\newblock In Sandy Irani, editor, {\em 61st {IEEE} Annual Symposium on Foundations of Computer Science, {FOCS} 2020, Durham, NC, USA, November 16-19, 2020}, pages 517--528. {IEEE}, 2020.

\bibitem[Ber23]{bernstein2023improved}
Aaron Bernstein.
\newblock Improved bounds for matching in random-order streams.
\newblock {\em Theory of Computing Systems}, pages 1--15, 2023.

\bibitem[BK96]{BK96}
Andr{\'{a}}s~A. Bencz{\'{u}}r and David~R. Karger.
\newblock Approximating \emph{s-t} minimum cuts in \emph{{\~{O}}}(\emph{n}\({}^{\mbox{2}}\)) time.
\newblock In Gary~L. Miller, editor, {\em Proceedings of the Twenty-Eighth Annual {ACM} Symposium on the Theory of Computing, Philadelphia, Pennsylvania, USA, May 22-24, 1996}, pages 47--55. {ACM}, 1996.

\bibitem[BK15]{BK02}
Andr{\'{a}}s~A. Bencz{\'{u}}r and David~R. Karger.
\newblock Randomized approximation schemes for cuts and flows in capacitated graphs.
\newblock {\em {SIAM} J. Comput.}, 44(2):290--319, 2015.

\bibitem[BSS12]{BSS12}
Joshua~D. Batson, Daniel~A. Spielman, and Nikhil Srivastava.
\newblock Twice-{R}amanujan sparsifiers.
\newblock {\em {SIAM} J. Comput.}, 41(6):1704--1721, 2012.

\bibitem[CCM16]{robust_2016}
Amit Chakrabarti, Graham Cormode, and Andrew McGregor.
\newblock Robust lower bounds for communication and stream computation.
\newblock {\em Theory of Computing}, 12(10):1--35, 2016.

\bibitem[CFPS19]{czumaj2019testable}
Artur Czumaj, Hendrik Fichtenberger, Pan Peng, and Christian Sohler.
\newblock Testable properties in general graphs and random order streaming.
\newblock {\em arXiv preprint arXiv:1905.01644}, 2019.

\bibitem[CGMV20]{chakrabarti2020vertex}
Amit Chakrabarti, Prantar Ghosh, Andrew McGregor, and Sofya Vorotnikova.
\newblock Vertex ordering problems in directed graph streams.
\newblock In {\em Proceedings of the Fourteenth Annual ACM-SIAM Symposium on Discrete Algorithms}, pages 1786--1802. SIAM, 2020.

\bibitem[CGP{\etalchar{+}}18]{CGP+18}
Timothy Chu, Yu~Gao, Richard Peng, Sushant Sachdeva, Saurabh Sawlani, and Junxing Wang.
\newblock Graph sparsification, spectral sketches, and faster resistance computation, via short cycle decompositions.
\newblock In Mikkel Thorup, editor, {\em 59th {IEEE} Annual Symposium on Foundations of Computer Science, {FOCS} 2018, Paris, France, October 7-9, 2018}, pages 361--372. {IEEE} Computer Society, 2018.

\bibitem[CKST19]{carlson2019optimal}
Charles Carlson, Alexandra Kolla, Nikhil Srivastava, and Luca Trevisan.
\newblock Optimal lower bounds for sketching graph cuts.
\newblock In {\em Proceedings of the Thirtieth Annual ACM-SIAM Symposium on Discrete Algorithms}, pages 2565--2569. SIAM, 2019.

\bibitem[CMP20]{CMP20}
Michael~B. Cohen, Cameron Musco, and Jakub Pachocki.
\newblock Online row sampling.
\newblock {\em Theory of Computing}, 16(15):1--25, 2020.

\bibitem[FHM{\etalchar{+}}20]{farhadi2020approximate}
Alireza Farhadi, Mohammad~Taghi Hajiaghayi, Tung Mah, Anup Rao, and Ryan~A Rossi.
\newblock Approximate maximum matching in random streams.
\newblock In {\em Proceedings of the Fourteenth Annual ACM-SIAM Symposium on Discrete Algorithms}, pages 1773--1785. SIAM, 2020.

\bibitem[FKM{\etalchar{+}}08]{feigenbaum2008}
Joan Feigenbaum, Sampath Kannan, Andrew McGregor, Siddharth Suri, and Jian Zhang.
\newblock Graph distances in the data-stream model.
\newblock {\em SIAM J. Comput.}, 38:1709--1727, 12 2008.

\bibitem[GKMS19]{gamlath2019weighted}
Buddhima Gamlath, Sagar Kale, Slobodan Mitrovic, and Ola Svensson.
\newblock Weighted matchings via unweighted augmentations.
\newblock In {\em Proceedings of the 2019 ACM Symposium on Principles of Distributed Computing}, pages 491--500, 2019.

\bibitem[GLPS12]{GLPS12}
Anna~C. Gilbert, Yi~Li, Ely Porat, and Martin~J. Strauss.
\newblock Approximate sparse recovery: Optimizing time and measurements.
\newblock {\em {SIAM} J. Comput.}, 41(2):436--453, 2012.

\bibitem[GM09]{doi:10.1137/07069328X}
Sudipto Guha and Andrew McGregor.
\newblock Stream order and order statistics: Quantile estimation in random-order streams.
\newblock {\em SIAM Journal on Computing}, 38(5):2044--2059, 2009.

\bibitem[HJJ{\etalchar{+}}14]{HendrickxPowerNetwork}
Julien~M. Hendrickx, Karl~Henrik Johansson, Raphaël~M. Jungers, Henrik Sandberg, and Kin~Cheong Sou.
\newblock Efficient computations of a security index for false data attacks in power networks.
\newblock {\em IEEE Transactions on Automatic Control}, 59(12):3194--3208, 2014.

\bibitem[HWK24]{hashemi2024weighted}
Diba Hashemi and Weronika Wrzos-Kaminska.
\newblock Weighted matching in the random-order streaming and robust communication models.
\newblock {\em arXiv preprint arXiv:2408.15434}, 2024.

\bibitem[JS18]{jambulapati2018efficient}
Arun Jambulapati and Aaron Sidford.
\newblock Efficient {\~o} (n/eps) spectral sketches for the laplacian and its pseudoinverse.
\newblock In {\em Proceedings of the Twenty-Ninth Annual ACM-SIAM Symposium on Discrete Algorithms}, pages 2487--2503. SIAM, 2018.

\bibitem[Kap21]{kapralov2021matching}
Michael Kapralov.
\newblock Space lower bounds for approximating maximum matching in the edge arrival model.
\newblock In {\em Proceedings of the Thirty-Second Annual ACM-SIAM Symposium on Discrete Algorithms}, SODA '21, page 1874–1893, USA, 2021. Society for Industrial and Applied Mathematics.

\bibitem[Kar00]{karger00}
David~R. Karger.
\newblock Minimum cuts in near-linear time.
\newblock {\em J. {ACM}}, 47(1):46--76, 2000.

\bibitem[KKS14]{kapralov2014approximating}
Michael Kapralov, Sanjeev Khanna, and Madhu Sudan.
\newblock Approximating matching size from random streams.
\newblock In {\em Proceedings of the twenty-fifth annual ACM-SIAM symposium on Discrete algorithms}, pages 734--751. SIAM, 2014.

\bibitem[KL02]{KL02}
David~R. Karger and Matthew~S. Levine.
\newblock Random sampling in residual graphs.
\newblock In John~H. Reif, editor, {\em Proceedings on 34th Annual {ACM} Symposium on Theory of Computing, May 19-21, 2002, Montr{\'{e}}al, Qu{\'{e}}bec, Canada}, pages 63--66. {ACM}, 2002.

\bibitem[KLM{\etalchar{+}}17]{KLM+17}
Michael Kapralov, Yin~Tat Lee, Cameron Musco, Christopher Musco, and Aaron Sidford.
\newblock Single pass spectral sparsification in dynamic streams.
\newblock {\em {SIAM} J. Comput.}, 46(1):456--477, 2017.

\bibitem[KMM12]{konrad2012maximum}
Christian Konrad, Fr{\'e}d{\'e}ric Magniez, and Claire Mathieu.
\newblock Maximum matching in semi-streaming with few passes.
\newblock In {\em International Workshop on Approximation Algorithms for Combinatorial Optimization}, pages 231--242. Springer, 2012.

\bibitem[KMM{\etalchar{+}}20]{spectral_stream_2020}
Michael Kapralov, Aida Mousavifar, Cameron Musco, Christopher Musco, Navid Nouri, Aaron Sidford, and Jakab Tardos.
\newblock Fast and space efficient spectral sparsification in dynamic streams.
\newblock In {\em Proceedings of the Thirty-First Annual ACM-SIAM Symposium on Discrete Algorithms}, SODA '20, page 1814–1833, USA, 2020. Society for Industrial and Applied Mathematics.

\bibitem[KN96]{kushilevitz_nisan_1996}
Eyal Kushilevitz and Noam Nisan.
\newblock {\em Communication Complexity}.
\newblock Cambridge University Press, 1996.

\bibitem[Kon18]{konrad2018simple}
Christian Konrad.
\newblock A simple augmentation method for matchings with applications to streaming algorithms.
\newblock In {\em 43rd International Symposium on Mathematical Foundations of Computer Science, MFCS 2018}, pages 74--1. Schloss Dagstuhl-Leibniz-Zentrum fur Informatik GmbH, Dagstuhl Publishing, 2018.

\bibitem[KS96]{KS96}
David~R. Karger and Clifford Stein.
\newblock A new approach to the minimum cut problem.
\newblock {\em J. {ACM}}, 43(4):601--640, 1996.

\bibitem[Mad10]{Mad10}
Aleksander Madry.
\newblock Fast approximation algorithms for cut-based problems in undirected graphs.
\newblock In {\em 51th Annual {IEEE} Symposium on Foundations of Computer Science, {FOCS} 2010, October 23-26, 2010, Las Vegas, Nevada, {USA}}, pages 245--254. {IEEE} Computer Society, 2010.

\bibitem[McG14]{McGregor2014}
Andrew McGregor.
\newblock Graph stream algorithms: A survey.
\newblock {\em SIGMOD Rec.}, 43(1):9–20, may 2014.

\bibitem[MMPS17]{monemizadeh2017testable}
Morteza Monemizadeh, Shan Muthukrishnan, Pan Peng, and Christian Sohler.
\newblock Testable bounded degree graph properties are random order streamable.
\newblock {\em arXiv preprint arXiv:1707.07334}, 2017.

\bibitem[MN20]{MN20}
Sagnik Mukhopadhyay and Danupon Nanongkai.
\newblock Weighted min-cut: sequential, cut-query, and streaming algorithms.
\newblock In {\em Proccedings of the 52nd Annual {ACM} {SIGACT} Symposium on Theory of Computing, {STOC} 2020, Chicago, IL, USA, June 22-26, 2020}, pages 496--509. {ACM}, 2020.

\bibitem[MP80]{MUNRO1980315}
J.I. Munro and M.S. Paterson.
\newblock Selection and sorting with limited storage.
\newblock {\em Theoretical Computer Science}, 12(3):315--323, 1980.

\bibitem[OS23]{ottolini2023concentration}
Andrea Ottolini and Stefan Steinerberger.
\newblock Concentration of hitting times in erd{\H{o}}s-r{\'e}nyi graphs.
\newblock {\em Journal of Graph Theory}, 2023.

\bibitem[PS18]{peng2018estimating}
Pan Peng and Christian Sohler.
\newblock Estimating graph parameters from random order streams.
\newblock In {\em Proceedings of the Twenty-Ninth Annual ACM-SIAM Symposium on Discrete Algorithms}, pages 2449--2466. SIAM, 2018.

\bibitem[PY19]{optimal_cycle_decomposition}
Merav Parter and Eylon Yogev.
\newblock Optimal short cycle decomposition in almost linear time.
\newblock In {\em 46th International Colloquium on Automata, Languages, and Programming, ICALP 2019}, Leibniz International Proceedings in Informatics, LIPIcs. Schloss Dagstuhl- Leibniz-Zentrum fur Informatik GmbH, Dagstuhl Publishing, 7 2019.

\bibitem[RSW18]{rsw18}
Aviad Rubinstein, Tselil Schramm, and S.~Matthew Weinberg.
\newblock Computing exact minimum cuts without knowing the graph.
\newblock In Anna~R. Karlin, editor, {\em 9th Innovations in Theoretical Computer Science Conference, {ITCS} 2018, January 11-14, 2018, Cambridge, MA, {USA}}, volume~94 of {\em LIPIcs}, pages 39:1--39:16. Schloss Dagstuhl - Leibniz-Zentrum f{\"{u}}r Informatik, 2018.

\bibitem[She09]{She09}
Jonah Sherman.
\newblock Breaking the multicommodity flow barrier for o(vlog n)-approximations to sparsest cut.
\newblock In {\em 50th Annual {IEEE} Symposium on Foundations of Computer Science, {FOCS} 2009, October 25-27, 2009, Atlanta, Georgia, {USA}}, pages 363--372. {IEEE} Computer Society, 2009.

\bibitem[SK23]{SugiuraRoadNetwork}
Satoshi Sugiura and Fumitaka Kurauchi.
\newblock Isolation vulnerability analysis in road network: Edge connectivity and critical link sets.
\newblock {\em Transportation Research Part D: Transport and Environment}, 119:103768, 2023.

\bibitem[SS11]{SS11}
Daniel~A. Spielman and Nikhil Srivastava.
\newblock Graph sparsification by effective resistances.
\newblock {\em {SIAM} J. Comput.}, 40(6):1913--1926, 2011.

\bibitem[ST04]{st04}
Daniel~A. Spielman and Shang{-}Hua Teng.
\newblock Nearly-linear time algorithms for graph partitioning, graph sparsification, and solving linear systems.
\newblock In L{\'{a}}szl{\'{o}} Babai, editor, {\em Proceedings of the 36th Annual {ACM} Symposium on Theory of Computing, Chicago, IL, USA, June 13-16, 2004}, pages 81--90. {ACM}, 2004.

\bibitem[ST11]{ST11}
Daniel~A. Spielman and Shang{-}Hua Teng.
\newblock Spectral sparsification of graphs.
\newblock {\em {SIAM} J. Comput.}, 40(4):981--1025, 2011.

\bibitem[SW15]{SW15}
Xiaoming Sun and David~P. Woodruff.
\newblock {Tight Bounds for Graph Problems in Insertion Streams}.
\newblock In Naveen Garg, Klaus Jansen, Anup Rao, and Jos{\'e} D.~P. Rolim, editors, {\em Approximation, Randomization, and Combinatorial Optimization. Algorithms and Techniques (APPROX/RANDOM 2015)}, volume~40 of {\em Leibniz International Proceedings in Informatics (LIPIcs)}, pages 435--448, Dagstuhl, Germany, 2015. Schloss Dagstuhl--Leibniz-Zentrum fuer Informatik.

\bibitem[WDL{\etalchar{+}}20]{wang2020efficient}
Xinjue Wang, Ke~Deng, Jianxin Li, Jeffery~Xu Yu, Christian~S Jensen, and Xiaochun Yang.
\newblock Efficient targeted influence minimization in big social networks.
\newblock {\em World Wide Web}, 23(4):2323--2340, 2020.

\bibitem[Wil16]{Lec13}
David~P. Williamson.
\newblock Lecture notes in spectral graph theory, October 2016.
\newblock \url{https://people.orie.cornell.edu/dpw/orie6334/Fall2016/lecture13.pdf}.

\bibitem[Zel11]{Zelke2011Intractability}
Mariano Zelke.
\newblock Intractability of min- and max-cut in streaming graphs.
\newblock {\em Information Processing Letters}, 111(3):145--150, 2011.

\end{thebibliography}
 \newpage
\appendix
\section{Lower Bound: Approximate Minimum Cut}\label{sec:lb}
This section gives an alternate proof of the randomized algorithm lower bound for finding an approximate minimum cut (\cref{thm:lower_bound}). 

\vspace{-5pt}\paragraph{\cref{thm:lower_bound}.}
\emph{Fix $\eps > 1/n$. Any randomized algorithm that outputs a $(1+\eps)$-approximation to the minimum cut of a simple, undirected graph in a single pass over a stream with probability at least $2/3$ requires $\Omega(n/\eps)$ bits of space.}
\vspace{5pt}
 
Notably, this shows that the result in \cref{thm:sparsifier} is optimal in space complexity up to polylogarithmic factors. In~\Cref{sec:min_cut}, we have shown this lower bound for randomized algorithm using the $k$-edge-connectivity problem in~\cite{SW15}.
Below we also provide an alternative self-contained construction for the randomized result using the Index communication problem.

We also provide a self-contained construction of the randomized algorithm lower bound result using the Index communication problem. Existing work has shown an $\Omega(n^2)$ bits of space lower bound for one-pass streaming algorithms that estimate the exact value of the minimum cut on undirected graphs \cite{Zelke2011Intractability}. This is done via a standard reduction from the Index problem in communication complexity. We use an extension of the same technique to prove a tight lower bound on approximate minimum cut streaming algorithms.

\subsection{Warm up: Exact Minimum Cut}
For the purpose of demonstration, we first give an overview of the proof for the exact minimum cut lower bound given by Theorem 1 of \cite{Zelke2011Intractability}. Recall that it is a standard known result in communication complexity that any randomized protocol solving the following Index problem on a $n$-bit binary vector with probability at least 2/3 requires $\Omega(n)$ bits of memory \cite{kushilevitz_nisan_1996} (see also \cref{lem:index}). Suppose that there exists a one-pass streaming algorithm $\mathcal{A}$ which computes the exact minimum cut value of any arbitrary simple unweighted graph of $n$ vertices using $o(n^2)$ bits of space with probability at least $2/3$. The reduction is given as follows. Alice has a bit vector $x$ of length $(n^2-n)/2$, which she uses to represent the upper half of the adjacency
matrix of the graph $G$. She instantiates $\mathcal{A}$ with input graph $H$ with $7n+1$ vertices, where the edges of $H$ will be determined later. The edges between the first $n$ nodes are the same as those in $G$. After feeding
the edges of $G$ into $\mathcal{A}$, Alice sends the memory configuration of $\mathcal{A}$ to Bob. Alice also sends the degree of each of the $n$ vertices in $G$ to Bob. In total, this is $O(n\log n)+o(n^2) = o(n^2)$ bits which are sent to Bob.

Bob has an index $i\in [1, (n^2-n)/2]$, corresponding to an edge that he wishes to determine whether it is in $G$. We call this potential edge $(a,b)$. Bob then adds edges to $H$ to construct a new graph $H^+$ under the following rules. He first adds edges into $H$ to construct two disjoint cliques $S$ and $T$ of $3n$ vertices each (here, the nodes in $S$ and $T$ are disjoint with the first $n$ nodes in $H$). Bob next adds edges from all vertices in $S$ to $a$ and $b$, and edges from all vertices in $T$ to all vertices in $V\setminus\{a,b\}$. We denote sets $L = S\cup\{a,b\}$ and $R=T\cup V\setminus\{a,b\}$.

Bob then adds edges from the last remaining vertex, which we call $c$, to an arbitrary set of $\deg_G(a) + \deg_G(b) - 1$ vertices in $V\setminus\{a,b\}$ (see \cref{fig:exact} for a complete diagram of the construction). He then queries $\mathcal{A}$ about the minimum cut value. The observation is that there are two possible minimum cuts: $C_1=(L,R\cup\{c\})$ with size $\deg_G(a)+\deg_G(b)-2$ if $ab$ is an edge in $G$, or $C_2=(L\cup R,\{c\})$ with size $\deg_G(a)+\deg_G(b)-1$ if $ab$ is not an edge in $G$. From this, Bob then determines the value $x[i]$ with $o(n^2)$ communication with probability at least $2/3$, a contradiction.

\begin{figure}
\includegraphics[width=10cm]{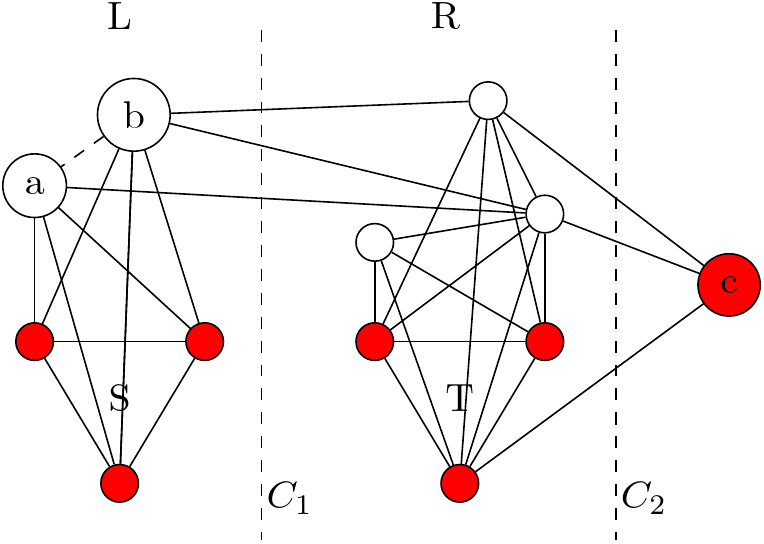}
\centering
\caption{Construction of graph $H^+$ for exact minimum cut from \cite{Zelke2011Intractability}. Red nodes are the ones that are added in addition to the original graph.}
\label{fig:exact}
\end{figure}

\subsection{Approximate Minimum Cut}
We are ready to give proof for \cref{thm:lower_bound}. First, we assume that $1/\eps < n$, as otherwise there already exists the known $\Omega(n^2)$ lower bound for minimum cut. Additionally, we assume $\eps = o(1)$, as a constant approximation for minimum cut answers whether a graph is connected or not, and there already exists a $\Omega(n)$ lower bound for connectivity \cite{feigenbaum2008}.

Suppose there exists a one-pass streaming algorithm $\mathcal{A}$ which computes a $(1+\eps)$-approximation to the minimum cut of any arbitrary simple unweighted graph on $n$ vertices using $o(n/\eps)$ bits of space and with probability at least $2/3$. Alice has a bit vector $x$ of length $s = \Theta(n/\eps)$, which she uses to represent the edges of a graph $G$, which is formed by $\eps n$ disjoint graphs each with $1/(4\eps)$ vertices. She instantiates $\mathcal{A}$ with an input graph $H$ with $(n + 6/\eps)/4 + 1$ vertices, in which the edges between the first $n$ nodes are the same as the edges in $G$. After feeding the edges of $G$ into $\mathcal{A}$, she then sends the memory configuration of $\mathcal{A}$ to Bob. Alice also sends the degree of each of the $n/4$ vertices in $G$ to Bob. In total this is $O(n\log (1/\eps))+o(n/\eps) = o(n/\eps)$ bits sent to Bob.

Bob has an index $i\in [1, s]$, which corresponds to an edge within $G$ that he wishes to determine whether it exists or not. Specifically, this edge belongs to one of the $\eps n$ disjoint subgraphs. Let us denote this subgraph as $G_i$.
Bob now constructs a new graph $H^+(\eps)$ under the following rules. He constructs two cliques $S$ and $T$, each of size $3/(4\eps)$. For subgraph $G_i$, Bob connects $S$, $T$, and vertex $c$ in the same manner as the exact minimum cut construction. For all other vertices in subgraphs other than $G_i$, Bob does not follow this construction and instead arbitrarily connects each one of them to all vertices in either $S$ or $T$. Denote the set of all subgraphs connected to $S$ as $G_S$, and $G_T$ similarly. See \cref{fig:approximate} for a complete diagram of the construction.

\begin{figure}
\includegraphics[width=10cm]{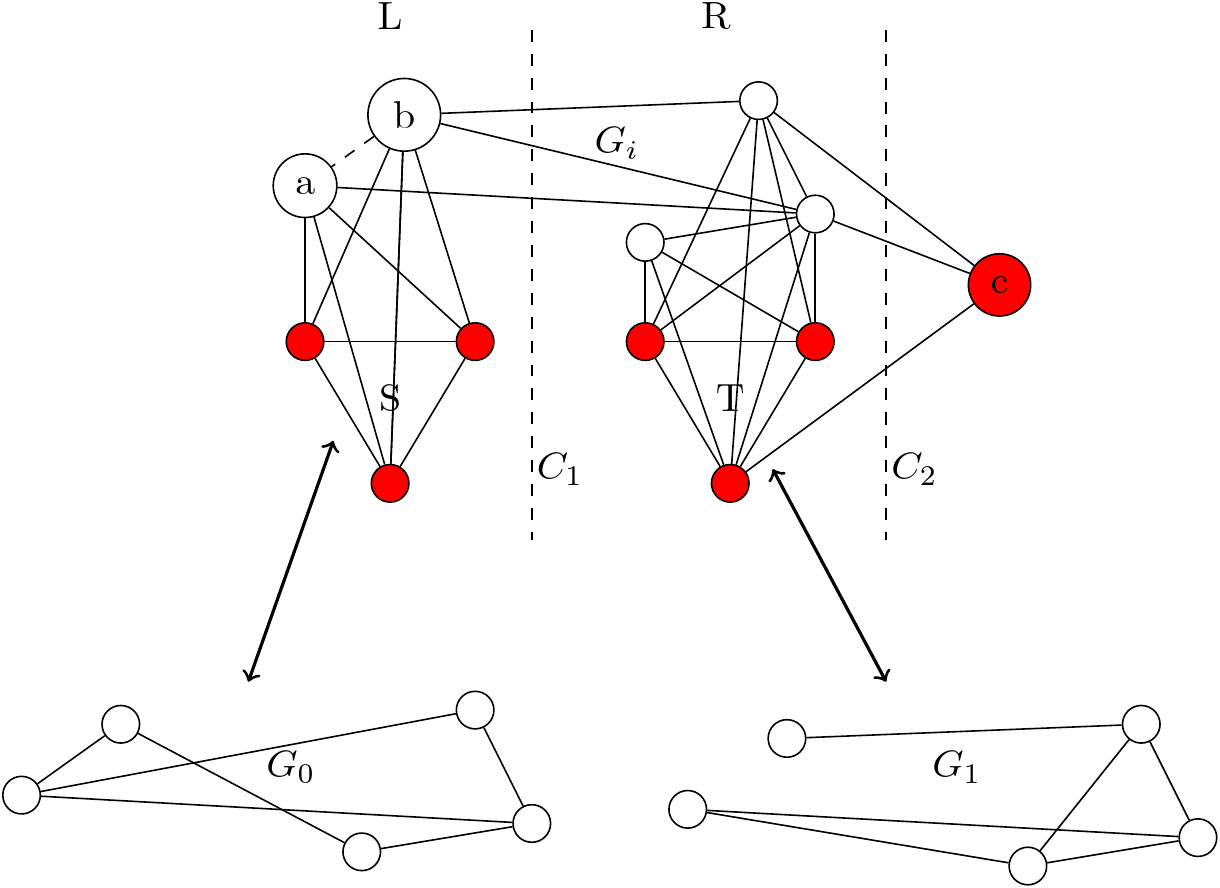}
\centering
\caption{Construction of graph $H^+(\eps)$ for approximate minimum cut. Red nodes are the ones that are added in addition to the original graph. $G_i$ represents the graph containing (potential) edge $ab$, and $G_0$ and $G_1$ represent the remaining graphs. The double arrow edges represent bicliques between graph $G_0\in G_S$ and $S$, and graph $G_1 \in G_T$ and $T$.}
\label{fig:approximate}
\end{figure}

Bob then queries $\mathcal{A}$ to get a $(1+\eps)$-approximation of the minimum cut value of $H^+(\eps)$. We denote sets $L = S\cup\{a,b\}$ and $R=T\cup V_i\setminus\{a,b\}$, where $V_i$ represents the vertices in graph $G_i$. Once again define the two candidate minimum cuts as $C_1=(L \cup G_S,R\cup G_T \cup \{c\})$ with size $\deg_G(a)+\deg_G(b)-2$ if $ab$ is an edge in $G$, and $C_2=(L\cup G_S \cup R \cup G_T,\{c\})$ with size $\deg_G(a)+\deg_G(b)-1$ if $ab$ is not an edge in $G$. The crucial observation here is that the two candidate minimum cuts are still $C_1$ and $C_2$, as shown by the following claim:
\begin{claim}
    All graph cuts in $H^+(\eps)$ besides $C_1$ and $C_2$ have size at least $3/(4\eps)-1$.
\end{claim}
\begin{proof}
    We will prove our claim by describing all cuts with size $<3/(4\eps)-1$. First, any cut that separates any two vertices within $L$ (or $R$, respectively) must have size at least $3/(4\eps)-1$, from an identical argument used by the construction in \cite{Zelke2011Intractability}. Additionally, any cut that splits any vertex within $G_S$ from $S$ ($G_T$ from $T$ respectively) must also have size at least $3/(4\eps)-1$, as every vertex in $G_S$ connects to all vertices in $S$. 

    Therefore, any cut with size $<3/(4\eps)-1$ cannot split any part of either $L \cup G_S$ or $R\cup G_T$. We can thus contract both sets into a single super-vertex with parallel edges. From this, the only vertices that remain are $c$ and the super-vertices of $L \cup G_S$ and $R\cup G_T$, and it is easy to verify the only remaining cuts of size $<3/(4\eps)-1$ are $C_1$ and $C_2$.
\end{proof}

Both cuts $C_1$ and $C_2$ have a maximum value of $<\deg_G(a) + \deg_G(b) < 1/(2\eps)$. Hence a $(1+\eps)$ approximation gives us the cut value with an additive error strictly less than 0.5, which is enough to distinguish whether $ab$ is in graph $G$.

\end{document}